\documentclass[letterpaper,12pt]{article}

\usepackage[left=1in, right=1in, top=1in, bottom=1in]{geometry}
\usepackage{amsfonts, dsfont}
\usepackage[pdftex]{graphicx}
\usepackage{amsmath,amssymb,amsthm}
\usepackage{microtype}
\usepackage{rotating}
\usepackage{pdflscape}
\usepackage{adjustbox}
\usepackage{xcolor}
\usepackage{epstopdf}
\usepackage{xr-hyper}
\usepackage{fancyhdr}
\usepackage{enumitem}
\usepackage{comment}
\usepackage{mathtools}
\usepackage[unicode=true,pdfusetitle,
 bookmarks=true,bookmarksnumbered=false,bookmarksopen=false,
 breaklinks=false,pdfborder={0 0 1},backref=false,colorlinks=true]
 {hyperref}
\hypersetup{colorlinks=true, linkcolor=red, citecolor=blue, filecolor=red}

\usepackage{setspace}
\makeatletter
\def\@setthanks{\vspace{-\baselineskip}\def\thanks##1{\@par##1\@addpunct.}\thankses}
\makeatother
\usepackage{natbib}
\usepackage{caption}
\usepackage{tabularx}
 \usepackage{float}
 \usepackage{booktabs}
 \usepackage{multirow}
\usepackage{bm}
\newtheorem{thm}{Theorem}[section]

\newtheorem{lem}{Lemma}[section]

\newtheorem{assum}{Assumption}
\newtheorem{prop}{Proposition}[section]
\numberwithin{equation}{section}

\newtheorem{remark}{Remark}
\newtheorem{innercustomgeneric}{\customgenericname}
\providecommand{\customgenericname}{}
\newcommand{\newcustomtheorem}[2]{%
  \newenvironment{#1}[1]
  {%
   \renewcommand\customgenericname{#2}%
   \renewcommand\theinnercustomgeneric{##1}%
   \innercustomgeneric
  }
  {\endinnercustomgeneric}
}
\numberwithin{equation}{section}
\newcustomtheorem{customthm}{Theorem}
\newcustomtheorem{customlem}{Lemma}
\newcustomtheorem{customassumption}{Assumption}
\newcustomtheorem{customprop}{Proposition}
\newcustomtheorem{customexample}{Example}
\newcustomtheorem{customdef}{Definition}
\newcustomtheorem{customcor}{Corollary}
\newcustomtheorem{customrem}{Remark}

\usepackage[title]{appendix}

\usepackage{todonotes}
\newcommand{\ind}{\perp\!\!\!\!\perp} 
\newcommand{\E}{\mathbb{E}}

\newcommand{\1}{\mathds{1}}
\newcommand{\R}{\mathbb{R}}
\newcommand{\I}{\mathcal{I}}
\renewcommand{\H}{\mathcal{H}}
\newcommand{\bb}{\mathbf{b}}

\renewcommand{\cite}{\citet}
\usepackage{todonotes}

\title{Endogenous Quantile Regression with Measurement Error in Dependent Variable}
\author{Xuanjing Su\thanks{Department of Economics, University of Wisconsin--Madison. Email: \href{mailto:xssu@wisc.edu}{xssu@wisc.edu}. I thank Bruce Hansen for his invaluable advice. I am also grateful for the inspiring discussions and helpful comments from Jack Porter, Xiaoxia Shi, Harold Chiang, Kohei Yata, Yong Cai, and participants at New York Camp Econometrics XX. Any remaining errors are my own.}}
\date{May 20, 2026}

\begin{document}
\onehalfspacing

\maketitle

\begin{abstract}
This paper studies quantile regression with an endogenous regressor and measurement error in the dependent variable. Standard quantile regression estimators ignoring these two elements can induce substantial bias. We adopt a control-function approach in a triangular system and show that the conditional quantile coefficient functions, together with all other distributional parameters, are nonparametrically identifiable. Building on this constructive identification result, we propose a two-step sieve ML estimator. The first step estimates the control function. The second step performs a sieve likelihood maximization that incorporates the generated control variable through copula weights. When the number of quantile grid knots grows at an appropriate speed, the estimator is consistent and asymptotically normal, permitting inference via bootstrap. Monte Carlo simulations demonstrate that the estimator markedly reduces bias relative to existing methods, confirming its effectiveness in settings with endogeneity and additive measurement error in the outcome.

\vspace{0.3in}
\noindent\textbf{Keywords:} Measurement error, quantile regression, control variable, copula.
\end{abstract}

\microtypesetup{
  activate=false,
  protrusion=false,
  expansion=false
}

\clearpage

\section{Introduction}
\label{sec:intro}

Quantile regression (QR) has been extensively studied in theoretical econometrics \citep{koenker1978regression, koenker2005quantile} and has gained substantial traction in empirical work. It helps capture the distributional effects of covariates on the outcome variable, not only at the center but also heterogeneously across different quantiles.

In observational data, however, two practical challenges frequently arise. First, variables of interest are often endogenous, rendering standard QR estimators inconsistent for marginal quantile effects. Second, variables used in empirical analyses are commonly subject to measurement error, especially in self-reported survey data, such as income \citep{bound1991extent, bound2001measurement}, health status \citep{bound1999dynamic}, or receipt of government transfers \citep{meyer2009under}, among many others. While left-hand-side errors in variables (LHS EIV) do not bias slope estimates in linear models, such robustness does not carry over to nonlinear models like quantile regression. \citet{hausman2021errors} examine LHS EIV in quantile regression under the assumption of independence between the treatment and unobserved heterogeneity. This paper extends their analysis to an endogenous quantile regression framework, allowing for dependence between individual heterogeneity and treatment selection, as is typical when the treatment arises as an equilibrium outcome of individual decisions.

We address endogeneity using a control-function approach within a triangular simultaneous equations model. The model consists of an outcome equation and a reduced-form first stage. The outcome equation takes the form of a linear conditional quantile specification with an additive measurement error in the dependent variable. This setup retains the structure in \citet{hausman2021errors}, while allowing a continuously distributed endogenous regressor to be correlated with the unobserved heterogeneity (i.e., the latent rank). The endogenous regressor is generated in the first-stage equation, where instruments are assumed independent of the reduced-form scalar disturbance. In this framework, the conditional distribution of the endogenous variable given the instruments, which coincides with the distribution of the first-stage disturbance, functions as the control variable \citep{imbens2009identification}. This control variable thereby restores conditional independence between the endogenous regressor and the latent rank in the outcome equation. 

Although the control-function approach effectively isolates the endogenous component, it also introduces additional nuisance parameters due to the dependence between the latent rank and the control variable. In our model, this dependence is naturally captured through a copula with uniform margins. Exploiting this marginal uniformity property, this paper establishes nonparametric identification of the structural quantile coefficient functions, as in the presence of both additive measurement error in the outcome variable and endogeneity. This extends the analysis in \citet{hausman2021errors} by accommodating endogeneity without imposing any parametric restrictions on the stochastic relationship between the latent rank and the first-stage disturbance.

For estimation, we propose a two-stage sieve maximum likelihood estimator (2SSMLE). We parametrize the nuisance distributions, while sieving over the main functional parameters, namely the quantile coefficient functions. The first stage constructs the control variable. The second stage applies a sieve MLE to recover the distributional effects conditional on both the endogenous regressor and the generated control variable. Under a $\sqrt{n}$-consistent parametric first stage and a mildly ill-posed second stage, the resulting estimator achieves pointwise asymptotic normality with convergence rate of $(n\kappa_{J_n})^{-1/2}$, where $\{\kappa_{J_n}\}$ denotes the decay rate of the minimum eigenvalue of the Hessian matrix as the sieve dimension $J_n\to\infty$, capturing the degree of ill-posedness. Both the convergence rate and the shape of the asymptotic variance are determined by the second stage, while the first-stage error propagates into the variance at first order. When the first stage is estimated nonparametrically, for example using a series estimator, an additional regularization bias arises due to series approximation error. Under an undersmoothing condition on the series dimension $K_n$, which ensures that the first stage approximation bias is asymptotically negligible relative to the second stage rate, the 2SSMLE retains asymptotic normality, thereby permitting inference by nonparametric bootstrap.

The proposed estimator applies to a wide range of empirical settings involving continuous treatments, where endogeneity can be addressed through control functions while the outcome variable is potentially subject to measurement error. A prominent example is the estimation of individual Engel curves, as studied in \citet{blundell2007semi} and \citet{imbens2009identification}. In this setting, the latent rank captures heterogeneous consumption preferences, total expenditure is endogenous and survey-reported expenditure shares are prone to misreporting. An empirical investigation along this line will be reported in future work. In this paper, Monte Carlo simulations show that neglecting measurement error can lead to substantial bias, even after correcting for endogeneity, and that the magnitude of this bias depends on both the scale and the shape of the error distribution.

This paper adds to the literature on measurement error in nonlinear models. Although mismeasurement in independent variables has received considerable attention, errors in the dependent variable remain relatively understudied. For discrete outcomes, \citet{hausman1998misclassification} analyze misclassification in discrete-response models. For continuous outcomes, \citet{abrevaya1999semiparametric} examines mismeasurement within a general linear index framework, while \citet{cosslett2004efficient} examine censored regressions. The present work specifically advances the literature on quantile regression. Most existing studies in this domain address measurement error in covariates and restore identification through instruments or repeated measures, including \citet{schennach2008quantile}, \citet{chesher2017understanding}, \citet{wei2009quantile},  \citet{firpo2017measurement}, and \citet{Song2026}. In contrast, left-hand-side measurement error in quantile regression has received far less attention. \citet{hausman2001mismeasured} demonstrates that such errors can produce substantial bias in quantile regression estimates when conditional heteroskedasticity is present. More recently, \citet{hausman2021errors} establishes that the parameters of a linear quantile regression model with left-hand-side measurement error remain identified and can be consistently estimated when the regressors are independent of the unobserved heterogeneity.\footnote{Alternatively, if there is a collection of sufficiently rich included exogenous variables $\mathbf{Z}$, such that the endogenous regressors $\mathbf{X}$ are independent of the latent rank $U$ conditional on $\mathbf{Z}$, then one can also resort to \citep{hausman2021errors}. However, such a set $\mathbf{Z}$ is always unavailable in empirical data, thus we need additional information, typically provided by either instrumental variables or by control functions.} Our paper extends their analysis by incorporating endogeneity and introducing a two-stage estimation procedure based on a control-function approach. It further highlights how copula structures can be used to integrate control variables into random-coefficient models. Relatedly, \citet{DotySong2023} consider nonparametric identification and estimation of conditional quantiles with LHS measurement error in a dynamic production function framework. Their setting imposes a form of Hicks-neutral productivity, under which the transitory shock and productivity are independent conditional on inputs. This structure effectively renders latent quantile heterogeneity exogenous after conditioning on the control function, placing the model closer in spirit to \citet{hausman2021errors} after conditioning. Finally, \citet{callaway2021distributional} study two-sided measurement error in QR models without endogeneity, and propose a three-step EM-type algorithm for a broad class of distributional effect parameters. 

The literature on identifying heterogeneous treatment effects with endogenous regressors is extensive and largely follows two approaches: instrumental variables methods, such as \citet{chernozhukov2005iv}, and control function methods, such as \citet{imbens2009identification} and \citet{newey1999nonparametric}. Both strategies, however, break down in the presence of left-hand-side measurement error, as the convolution of the error with unobserved heterogeneity distorts the observed conditional quantiles. We address this identification challenge by combining a control variable with the scalar unobservable through a copula representation. Copulas have been widely used in econometrics to model dependence structures. For instance, among many others, \citet{han2017identification} employ single-parameter copulas to characterize dependence among unobservables, and \citet{callaway2021distributional} also adopt copula to model the joint distribution of treatment and outcome. In this paper, we impose no parametric restrictions on the copula in identification. Instead, identification is restored through large-support instruments and the monotonicity inherent in the linear quantile outcome specification.

The rest of the paper is organized as follows. Section \ref{sec:modelid} introduces the model specification and develops the nonparametric identification results. Section \ref{sec:est} presents the two-stage sieve ML estimator. Section~\ref{sec:asymp} derives its asymptotic properties and establishes the validity of the bootstrap for inference. Section~\ref{sec:sim} reports Monte Carlo evidence. Section~\ref{sec:con} concludes the study and outlines directions for future research. Appendix~\ref{sec:appendix} collects all proofs, additional simulation results, and implementation details of the estimator.

We adopt the following notation: We use upper-case Latin letters for random variables and the corresponding lower-cases for their realizations. Bold Latin letters denote matrices. For a random vector $X$ with support $\mathcal{X}$, $X_p$ denotes its $p$th dimension, and $X_{-p}$ denotes the subvector of $X$ corresponding to all but the $p$th dimension. $\E_n(X_i)=\frac1n\sum_{i=1}^n X_i$ denotes the empirical analogue of $\E[X_i]$. We write $F(\cdot)$ for a CDF, $F_{\cdot|\cdot}(\cdot|\cdot)$ for a conditional CDF, and $f(\cdot)$ for a PDF. For a parameter space $\Theta$, $\Theta^{\circ}$ denotes its interior, and $\mathcal{N}(\theta)\subseteq\Theta$ is a neighborhood of $\theta\in\Theta$. For any $\tau\in(0,1)$, the check function is $\rho_\tau(w)=w(\tau-\mathds{1}(w<0))$.


\section{Model and Identification}
\label{sec:modelid}

Let $Y^*$ be the latent outcome and $Y$ be the observed outcome measured with error. Let $X$ be a random vector with dimension $d_x$ and support $\mathcal{X}$. The model we consider has an outcome equation
\begin{equation}\label{e1}
  Y=Y^*+\varepsilon=X'\beta_0(U)+\varepsilon,
\end{equation} where $\varepsilon$ is a scalar measurement error. We adopt a classical mismeasured assumption: $\varepsilon$ is mean-zero, i.i.d., and independent of $Y^*$ and $X$. All unobserved heterogeneity in $Y^*$ is absorbed by a scalar rank $U\sim\mathcal{U}[0,1]$, and $x'\beta_0(\cdot)$ is monotonically increasing for each value $x\in\mathcal{X}$. When $X$ and $U$ are independent, the $\tau$th conditional quantile of $Y^*$ on $X=x$ is $Q_{Y^*|X}(x;\tau)=x'\beta_0(\tau)$, so that $U$ naturally represents the unobserved quantile level of $Y^*$ conditional on $X$. 

In practice, however, $U$ may be correlated with $X$, which often incorporates equilibrium outcomes partially determined by individual heterogeneity. Consider a triangular system with a single continuously distributed endogenous variable $X_1$ included in $X=(X_1,Z_1')'$, where $Z_1$ is a vector of exogenous covariates. In general, if $X_1$ becomes independent of $U$ after conditioning on a sufficiently rich set $Z_1$, then endogeneity ceases to be a concern, and $U$ retains the same interpretation. Such conditioning variables are rarely available in practice, and additional information is therefore required.

We obtain this additional information from a control variable. Let $Z=(Z_1',Z_2')'$, where $Z_2$ contains the excluded instrumental variables. The reduced form for $X_1$ is given by
\begin{equation}\label{e2}
	X_1=h(Z,\eta),
\end{equation} where $h(Z,\eta)$ is strictly monotonic in the scalar disturbance $\eta$ with probability one (w.p.1), and $Z$ is independent of $(\eta,U)$. While we allow for discrete instruments in $Z_2$, $\eta$ is a continuously distributed scalar with a strictly increasing CDF on its support. In this setting, \citet{imbens2009identification} demonstrate that the uniformly distributed 
\begin{equation}\label{e3}
	V=F_{X_1|Z}(X_1,Z)=F_\eta(\eta)
\end{equation} serves as a control variable: $X$ and $U$ are independent conditional on $V$.\footnote{Our framework is not restricted to a single endogenous regressor. Appendix~\ref{subsec:multi} presents an extension to multiple endogenous variables within a triangular system. Identification and estimation proceed analogously by employing a vector of control variables $\mathbf{V}$ and a corresponding multivariate copula $F_{U|\mathbf{V}}$.}

In this framework, the coefficient function $\beta_0(\cdot)$ admits multiple interpretations. One perspective, based on the potential outcome framework with a continuous treatment (e.g., \citet{chernozhukov2005iv, d2023nonparametric,callaway2024difference}), views $\beta_{0,1}(\tau)$ as the $\tau$-th conditional quantile treatment effect (QTE) of $X_1$, that is, $\partial_{x_1} q_{Y^*}(\tau,x_1,z_1)=\beta_{0,1}(\tau)$, holding $X_1=x_1$ and conditioning on $Z_1=z_1$. Alternatively, in a triangular simultaneous equations model, the quantile structural function (QSF) $q_{Y^*}(\tau,x)$, defined as the $\tau$-th quantile of $Y^*_x$, equals $x'\beta_0(\tau)$ under a monotonicity condition \citep{lehmann2006nonparametrics,imbens2009identification}. In this representation, $\beta_0(\cdot)$ is the quantile coefficient function that parameterizes the structural quantile response of $Y^*$ to covariates.

However, in the presence of both endogeneity and measurement error, pointwise identification of $\beta_0(\tau)$ fails. Consequently, simple quantile regression of $Y$ on $X$ may lead to significant bias. Namely, note that the minimization problem 
\begin{align*}
	\beta(\tau)\in \arg\min_b \E[\rho_\tau(y-x'b)]
\end{align*} is generally no longer minimized at the true $\beta_0(\tau)$. First, endogeneity shifts the $\tau$th conditional quantile of $Y^*|X=x$ to $x'\beta_0(\tilde\tau)$, where $\tilde\tau=F^{-1}_{U|V}(\tau)$. Moreover, measurement error further moves the $\tau$th conditional quantile of $Y|X=x$ away from $x'\beta_0(\tilde\tau)$ through convolution with the distribution of $\varepsilon$. Hence, instead of assigning a moment condition for each $\beta_0(\tau)$, we characterize $\beta_0(\cdot)$ using the conditional likelihood function given in (\ref{e5}).

\begin{prop}\label{prop1}
	\textbf{Main Statistical Implication}:\\
	Under the maintained conditions (i) $\varepsilon\ind (Y^*,X,Z)$ and (ii) $X\ind U\mid V$, the conditional distribution $F_{Y|X,V}$ and the conditional density $f_{Y|X,V}$ satisfy
\begin{align}
		F_{Y\mid X,V}(y \:| x,v)
        &= \int_0^1 F_\varepsilon\!\big(y - x'\beta_0(u)\big)\, F_{U| V}(du| v),\label{e4}\\ 
		f_{Y\mid X,V}(y \:| x,v)
        &= \int_0^1 f_\varepsilon\!\big(y - x'\beta_0(u)\big)\, f_{U|V}(u|v)\, du.\label{e5}
	\end{align}
\end{prop} Technically, proposition \ref{prop1} shows that the observed conditional distribution $F_{Y|X,V}$ is a convolution of $F_\varepsilon$ and $F_{U|V}$. A formal derivation is provided in Appendix \ref{pf_prop1}. In the absence of endogeneity (i.e., when $U\ind V$), $F_{U|V}(du|v)=F_U(du)=du$, and \eqref{e4} reduces to the representation in \citep{hausman2021errors}:
\begin{align*}
	F_{Y\mid X}(y \:| x)
        &= \int_0^1 F_\varepsilon\!\big(y - x'\beta_0(u)\big)\, du.
\end{align*} With an endogenous regressor, in contrast, $F_\varepsilon$ is weighted by the copula $F_{U|V}$, introducing an additional nuisance component, which captures the dependence between the latent heterogeneity $U$ and the endogenous variation in $X$. While \citet{hausman2021errors} establish nonparametric identification of $(\beta_0(\cdot),f_\varepsilon)$, we proceed to show that $(\beta_0(\cdot),f_\varepsilon,f_{U|V})$ can also be nonparametrically identified under the following assumptions in Section~\ref{subsec:npid}.

\subsection{Nonparametric Identification}
\label{subsec:npid}

\begin{assum} Restrictions on $X$:\label{npx}
  \begin{enumerate}[label=(\alph*)]
    \item $\E[XX']$ is non-singular, and $\E\|X\|^{2+\epsilon}<\infty$ for some $\epsilon>0$.
    \item Let $X=(X_p,X_{-p})$, where $X_p$ is a scalar  variable and $X_{-p}$ denotes the remaining components. For every $X_{-p}$ with strictly positive density or mass, there exists an open neighborhood of $X_p$ within which $X_p|X_{-p}$ is continuously distributed with strictly positive probability density.\footnote{This $X_p$ could naturally be our endogenous continuous variable $X_1$, but not necessary.}
  \end{enumerate}
\end{assum}

\begin{assum} Restrictions on $\beta(\cdot)$:\label{npb}
\begin{enumerate}[label=(\alph*)]
	\item $\beta_0$ is in the functional space $M_{B_1\times\dots\times B_{dx}}\equiv\{\beta=(\beta_1,\dots,\beta_{dx})':[0,1]\to B_1\times\dots\times B_{dx}\}$, where $B_k\subseteq \R$ denotes closed and uniformly bounded interval for each $k\in\{1,\dots,dx\}$, and {$g_x(u)\equiv x'\beta(u):[0,1]\to\R$ is monotonically increasing in $u$ for any $ x\in\mathcal{X}$}.
	\item For each $k\in\{1,\dots,dx\}$, $\beta_{0,k}$ is in the subspace $M_k\equiv\{e_k'\beta:\beta\in M_{B_1\times\dots\times B_{dx}}\}$, where $e_k\in\R^{dx}$ is the canonical basis vector. Moreover, the coefficient function of $X_p$ satisfies {$\beta_{0,p}\in M_p\subseteq \mathcal{C}([0,1])\cap\{\beta_p:\text{strictly monotonic}\}$}.
\end{enumerate}
\end{assum}

\begin{assum} Restrictions on EIV:\label{npe}
  \begin{enumerate}[label=(\alph*)]
    \item The unobserved $\varepsilon$ is continuously distributed with density $f_{\varepsilon}$ and $\varepsilon\ind (X,Z,Y^*)$.
    \item $E[\varepsilon]=\int e f_\varepsilon(e)de=0$ and $\phi_\varepsilon(s)=\E[\exp(is\varepsilon)]\neq0$ for all $s\in\R$.
    \item There exists a constant $C>0$ such that for all $q>0$, $\E[|\varepsilon^q|]<q!\cdot C^q$. 
  \end{enumerate}
\end{assum}

\begin{assum} Restrictions on copula: \footnote{A bivariate copula $C:[0,1]^2\to[0,1]$ is the joint distribution function of two variables with uniform margins. In our setting, $F_{U,V}$ inherently obeys such structure by the Sklar's theorem and continuity of $X_1$.}\label{npc}
\begin{enumerate}[label=(\alph*)]
	\item $U,V\sim \mathcal{U}[0,1]$. $(U,V)$ jointly follow some copula distribution $C(u,v)$ with continuous joint density $f_{U,V}$, equal to its conditional density $f_{U|V}$.\footnote{ Since $U$ and $V$ both have uniform margins, we have $f_{U,V}(u,v)=f_{U|V}(u|v)f_V(v)=f_{U|V}(u|v)$.}
	\item (Common support) For almost all $x\in\mathcal{X}$, the support of $V|X=x$ equals that of $V$.
\end{enumerate}
\end{assum}

With the above conditions on the parameters, covariates, and the two distributions, we state our main nonparametric identification result as follows.

\begin{thm} \textbf{Nonparametric Global Identification}\label{thm1}\\
	Denote $S$ as the collection of parameters $(\beta(\cdot),f_\varepsilon,f_{U|V})$ satisfying Assumptions \ref{npb}, \ref{npe}, and \ref{npc} respectively. Then under Assumptions \ref{npx}-\ref{npc}, for any set $(\beta(\cdot),f_\varepsilon,f_{U|V})$ that generates the same observed distribution $F_{Y|X,V}$ as the true parameter set $(\beta_0(\cdot),f_{0,\varepsilon},f_{0,U|V})$, it must be that $\beta_0(u)=\beta(u)$ for all $u\in(0,1)$, $f_{0,\varepsilon}(e)=f_\varepsilon(e)$ almost everywhere (a.e.) for all $e\in\R$, and $f_{0,U|V}(u,v)=f_{U|V}(u,v)$ a.e. for all $(u,v)\in[0,1]\times[0,1]$. 
\end{thm}

See Appendix~\ref{pf_thm1} for the proof. As suggested by \citet{hausman2021errors}, monotonicity of $x'\beta_0(\cdot)$, together with independence of $\varepsilon$, is critical for securing nonparametric identification. Although the unobservable measurement error effectively blur each quantile curve through a common kernel, the quasi-linear quantile structure, combined with the presence of one continuous regressor that enters strictly monotonically, transforms this into an identifiable deconvolution problem. In our setting, the common support assumption allows us to integrate the copula weights $f_{U|V}$ over $v\in(0,1)$. Because both $U$ and $V$ have uniform marginal distributions, this integration effectively collapses the problem to the exogenous case. In other words, identification of $(\beta_0(\cdot),f_\varepsilon)$ does not depend on specifying a parametric form for $F_{U|V}$. Once $\beta_0(\cdot)$ is identified, the copula $F_{U|V}$ can then be recovered.

\begin{remark}
	The common support assumption in Assumption~\ref{npc}(b) parallels that in \citep{imbens2009identification}. In a triangular system, where $V| X=x=(x_1,z_1)$ equals $F_{X_1|Z}(x_1,z_1,Z_2)$, this assumption entails that $Z_2$ not only shifts $V$ but also varies sufficiently to span its full support for each $x$. This is strong: it combines a rank-type condition with a full-support condition and implicitly entails at least one continuous component in $Z_2$. In \citep{imbens2009identification}, failure of this assumption still allows one to bound the object of interest. Here, by contrast, we rely on it to recover the uniform marginal of $U$ via integration over the full support of $v$; hence, we cannot allow $\mathcal{V}(x)\subsetneq[0,1]$. Nonetheless, as long as the assumption holds on a subset $\mathcal{X}_c=\{x\in\mathcal{X}:\mathcal{V}(x)=\mathcal{V}=[0,1]\}$, identification of $(\beta_0(\cdot),f_\varepsilon,f_{U|V})$ remains feasible by restricting $\mathcal{X}$ to $\mathcal{X}_c$. Empirically, the common support condition can be in principle checked by examining the data, thus informing the choice of $\mathcal{X}_c$. In that case, $\beta_0(\cdot)$ remains identifiable and consistently estimable.
\end{remark}


\section{Estimation}
\label{sec:est}

Based on the nonparametric identification result, this section describes our estimator. We adopt a two-step estimation strategy based on independent and identically distributed (i.i.d.) data $\{Y_i,X_i,Z_i\}_{i=1}^n$. The first step estimates the control variable $V_i$, yielding $\hat V_i$. The second step plugs in the generated variable $\hat V_i$ to construct a sieve maximum likelihood estimator (SMLE). While Theorem~\ref{thm1} demonstrates nonparametric identification of $(\beta_0(\cdot),f_{0,\varepsilon},f_{0,U|V})$, for estimation, we assume that $f_{0,\varepsilon}$ and $f_{0,U|V}$ are known up to a finite set of parameters,\footnote{In practice, the number of mixture components in each distribution can be increased flexibly to ensure adequate approximation.} and construct a semi-nonparametric estimator targeting $\beta_0(\cdot)$. Specifically, we impose the following assumptions on $(f_{0,\varepsilon},f_{0,U|V})$ in addition to Assumptions~\ref{npe} and \ref{npc}.

\begin{assum}\label{parpro}
	Parametric properties of EIV and copula: Hereafter, we use $f_c$ to denote the conditional copula  density $f_{U|V}$. The pair $(f_\varepsilon,f_c)$ is parametrized as $(f_\varepsilon(\cdot;\sigma), f_c(\cdot;\gamma))$, with true values $(f_{0,\varepsilon},f_{0,c})$ being $(f_\varepsilon(\cdot;\sigma_0), f_c(\cdot;\gamma_0))$.
	\begin{enumerate}[label=(\alph*)]
		\item $\sigma_0\in\Sigma^{\mathrm{o}}\subseteq\R^{d_\sigma}$, $\gamma_0\in\Gamma^{\mathrm{o}}\subseteq\R^{d_\gamma}$, where both $\Sigma$ and $\Gamma$ are compact.
		
		\item $f_{\varepsilon}$ is uniformly bounded and strictly positive on the real line. It is also twice differentiable in both $\varepsilon$ and $\sigma$ a.e. with uniformly bounded derivatives up to the second order.	
			
		\item There exists a sequence $\tau_n \to 0$ with $\tau_n \in (0,1/2)$ such that $(U,V)\in \mathcal{T}_n := [\tau_n,1-\tau_n]^2$ with probability approaching one (w.p.a.1). On $\mathcal{T}_n$, for all $\gamma\in\Gamma$, $f_c(\cdot,\cdot;\gamma)$ is twice continuously differentiable, and its derivatives are $L^p$-integrable up to the second order. Moreover, there exists positive $c_n \to 0$ such that $\inf_{\gamma\in\Gamma}\inf_{(u,v)\in\mathcal T_n} f_c(u,v;\gamma) \ge c_n$.

		\item There exists a uniform constant $\bar{C}>0$ such that, for all $\sigma\in\Sigma$ and 
      $\gamma\in\Gamma$, $\E[|\log f_\varepsilon(\varepsilon;\sigma)\cdot f_c(u|v;\gamma)|]<\bar{C}$. Moreover, there exists $C>0$ such that, at the true value $(\beta_0(\cdot),\sigma_0,\gamma_0)$, $\E\left[\|x\frac{\partial_\varepsilon f_\varepsilon(y-x'\beta_0(\tau);\sigma_0)}{f(y|x,v;\beta_0)}\|^4\right]<C$ and $\E\left[ \|\frac{\int f_\varepsilon(y-x'\beta_0(u);\sigma_0)\partial_\gamma f_c(u|v;\gamma_0)du}{f(y|x,v;\beta_0)}\|^4\right]<C$.	
		
		\item For any $\sigma\in\Sigma$, $l>0$, $\int_{-l}^l|\phi_\varepsilon(s;\sigma) - \phi_\varepsilon(s;\sigma_0)|^2ds\ge C_l\|\sigma-\sigma_0\|_2^2$ for some constant $C_l>0$, where $\phi_\varepsilon(s;\sigma)$ denotes the characteristic function of $\varepsilon$ given density $f_\varepsilon(\,\cdot\,;\sigma)$.
		
		\item For all $\beta\in\mathcal{N}(\beta_0)$ w.p.a.1, $\gamma\in\Gamma,l>0$, and some constant $c_l>0$, 
		\begin{align*}
			\E_{x,v}\left[\int_{-l}^l \left(\int_0^1\exp(isx'\beta(u))(f(u|v;\gamma)-f(u|v;\gamma_0))du\right)^2ds\right]\ge c_l\|\gamma-\gamma_0\|_2^2,
		\end{align*} that is, $\E_{x,v}[\int_{-l}^l |\phi_{x'\beta,\gamma}(s|x,v)-\phi_{x'\beta,\gamma_0}(s|x,v)|^2 ds]\ge c_l\|\gamma-\gamma_0\|_2^2.$
	\end{enumerate}
\end{assum}

Assumption~\ref{parpro} (b) holds for a broad class of mean-zero error distributions in the exponential family. Condition (c) accommodates many commonly used copulas, including Gaussian, Student-$t$, Clayton, and Frank. Specifically, the trimming sequence $\{c_n\}$ imposes a lower bound condition in each finite sample as the copula density approaches the boundary of the unit square. Condition (d) ensures uniformly bounded fourth moments of the score functions, so that the information matrix is locally nonsingular when applying the triangular central limit theorem. Finally, following \citet{hausman2021errors}, conditions (e)-(f) ensure local identification of $(\sigma,\gamma)$ through characteristic-function-based transforms.

Under such parametrization, we adopt the following parameter space and sieve spaces for approximation. Let $\theta_0:=(\beta_0(\cdot),\sigma_0,\gamma_0)\in\Theta_0=M\times\Sigma\times\Gamma$, for $(M,\Sigma,\Gamma)$ satisfying Assumptions~\ref{npb}-\ref{parpro}. Moreover, for each $k\in\{1,\dots,d_x\}$, let $M_k\subseteq \Lambda_{c_k}^{p_k}([0,1])$ with $p_k>\frac12$.\footnote{Following \citet{chen2007large}, $\Lambda^p_c(\mathcal{X})$ denotes the class of all $p$-smooth real-valued functions on $\mathcal{X}$. A real-valued function $h$ is $p$-smooth $(p=m+\gamma)$ if it is $m$ times continuously differentiable on $\mathcal{X}$ and its differential operator $D^\alpha h$ satisfies a H\"{o}lder condition with exponent $\gamma$ for all $\alpha$ with $[\alpha]=m$ and a dilation $c$.} For any $\theta,\tilde\theta\in \Theta_0$, define the metric $d(\theta,\tilde\theta):=\sqrt{\|\beta-\tilde\beta\|_{2,2}^2+\|\sigma-\tilde\sigma\|_2^2+\|\gamma-\tilde\gamma\|_2^2}$, where $\|\cdot\|_2$ denotes the Euclidean norm and $\|\cdot\|^2_{2,2}$ is the $l_2$-norm of $\|\beta(u)-\tilde\beta(u)\|^2_2$ across $u\in[0,1]$. The sieve spaces are defined as follows.

\begin{assum}\label{sieve}
	Sieve spaces: Define $\theta:=(\beta(\cdot),\sigma,\gamma)\in\Theta_n=M_n\times\Sigma\times\Gamma$, where:	\begin{enumerate}[label=(\alph*)]
		\item For all $x\in\mathcal{X}$, $u\mapsto x'\beta(u)$ is monotonically increasing on $u\in[0,1]$.
		\item Each component $M_{n k}$ is a spline space $\mathrm{Spl}(r,\tilde{J}_n)$ with $r\ge \max_{k}[p_k]+1$, so that 
		 \begin{align*}
		 	\beta_k(u)=\sum_{j=0}^{r-1} b_{j,k}u^j +\sum_{j=1}^{\tilde{J}_n}b_{r+j,k}(u-\tau_j)^{r-1}\1\{u\ge \tau_j\}=\sum_{j=0}^{r+\tilde{J}_n}b_{j,k}S_j(u),
		 \end{align*} where $0=\tau_0<\tau_1<\cdots<\tau_{\tilde{J}_n}<\tau_{\tilde{J}_{n+1}}=1$ partition $[0,1]$ into $\tilde{J}_n+1$ subintervals.
	\end{enumerate}
\end{assum} 

By definition, the sieve spaces satisfy $\Theta_n\subseteq\Theta_{n+1}\subseteq \Theta_0$ with each $\Theta_n$ compact under the metric $d(\cdot,\cdot)$ and $\inf_{\theta\in\Theta_n} d(\theta,\theta_0)\to0$. Define $\theta_J^*:=\pi_n\theta_0=(\beta_J^*(\cdot),\sigma_0,\gamma_0)$, where $\pi_n$ denotes the projection mapping from $\Theta_0$ to $\Theta_n$. Denote the sieve dimension as $J_n=\tilde{J}_n+r+1$. For each $k=1,\dots,d_x$, the sieve approximation is $\beta_k(u)=b_k' S(u)$, $b_k\in\R^{{J}_n}$, with basis vector $S(u)=(S_1(u),\dots,S_{{J}_n}(u))'$. Let $\bb := (b_1',\dots,b_{d_x}') \in \mathbb{R}^{d_x \times{J}_n}$ denote the stacked coefficient vectors, then $x'\beta(u) = \sum_{k=1}^{d_x} x_k b_k' S(u)=x'\cdot \bb \cdot S(u)$. The directional derivative of $f(y|x,v;\theta)$ w.r.t $b_k$ along direction of $S(u)$ is given by
\begin{align}\label{hadamard}
	\partial_{b_k}f(y|x,v;\theta):= \int_0^1 (-x_k)\cdot \partial_\varepsilon f(y-x'\bb S(u);\sigma)f_c(u|v;\gamma)S(u)du.
\end{align} Let $\partial_{{b}} f$ denote the stacked gradient in $\mathbb{R}^{d_x {J}_n}$. The score function is given by
\begin{align*}
	\partial_\theta \log f(y|x,v;\theta)=
\frac{1}{f(y|x,v;\theta)}
\begin{pmatrix}
\partial_{b} f \\[0.5ex]
\partial_\sigma f \\[0.5ex]
\partial_\gamma f
\end{pmatrix},
\end{align*} with the corresponding information matrix defined as $\mathcal{I}:=\E[(\partial_\theta\log f)(\partial_\theta\log f)']$ and the Hessian matrix as $\mathcal{H}:=\E[\partial_{\theta\theta'}\log f]$. 

Our estimation proceeds in two steps. The first step constructs the control variable $V=F_{X_1|Z}=F_\eta$. Following \citet{hahn2013asymptotic}, we consider both  parametric and nonparametric approaches for this step, which clarifies how first-stage estimation error propagates into the second-stage criterion, even when the first-stage estimator achieves $\sqrt{n}$-consistency. In the parametric specification, suppose
\begin{align}\label{first_p}
	X_{1}=Z'\pi_0+\eta,\qquad \eta \ind Z, \quad \eta \sim F_\eta,
\end{align} where $F_\eta$ is known. Let $\hat\pi$ denote the OLS estimator of $\pi_0$, satisfying $\hat\pi-\pi_0=O_p(n^{-1/2})$. The control variable is then constructed as
\begin{align}\label{v_p}
	\hat{V}_i\equiv F_\eta(X_{1i},Z_i;\hat\pi)=F_\eta(X_{1}-Z_i'\hat\pi).
\end{align} Such parametric specifications are often justified when researchers are willing to impose structural assumptions; see, for example, \citet{petrin2010control}. More generally, we estimate the control variable nonparametrically by approximating the conditional distribution function $F_{X_1|Z}$. Following \citet{imbens2009identification}, we employ series estimators, which are less sensitive to edge effects that often present in triangular models.\footnote{In particular, the joint density of $(X_1,V)$ may approach zero near the boundary of the support of $V$, thereby limiting the amount of information available in the tails \citep{imbens2009identification}.} Let $p^{K_n}(z) = (p_{1K_n}(z),\dots,p_{K_nK_n}(z))'$ denote a $K_n\times1$ vector of spline basis functions, then
\begin{gather}\label{v_np}
\hat{V}_i\equiv \hat{F}_{X_1\mid Z}(X_{1i},Z_i)=p^{K_n}(Z_i)' \hat a_n(X_{1i}), \\
\hat a_n(x) =\left(\sum_{j=1}^n p^{K_1}(Z_j)p^{K_1}(Z_j)' \right)^{-} \sum_{j=1}^n p^{K_1}(Z_j)\mathbf{1}\{X_{1j}\le x\},\nonumber
\end{gather} where $A^{-}$ denotes a generalized inverse of the matrix $A$.

The second stage implements sieve maximum likelihood estimation by replacing the unobserved control variable $V$ in the conditional log-likelihood with its estimate $\hat{V}$. Denote $W_i=(Y_i,X_i',V_i)'$ and $\hat{W}_i=(Y_i,X_i',\hat{V}_i)'$. Let $l(W_i;\theta)=\log f(Y_i|X_i,V_i;\theta)$. The expected conditional log-likelihood function is 
\begin{equation}
Q(\theta)=\E[l(W_i;\theta)]=\E\left[\log \int_0^1 f_\varepsilon(Y_i-X_i'\beta(u);\sigma)f_c(u|V_i;\gamma)du\right].
\end{equation} By Theorem~\ref{thm1}, the true parameter $\theta_0$ uniquely maximizes $Q(\theta)$ over $\Theta_0$. Let $\hat{Q}(\theta)=\E[l(\hat{W}_i;\theta)]$, with $Q_n(\theta)$ and $\hat{Q}_n(\theta)$ be the corresponding empirical counterparts. The feasible estimator $\hat\theta_n$ is then defined to maximize $\hat{Q}_n(\theta)$ over $\Theta_n$, i.e.,\footnote{For any compact $\hat{\mathcal{V}}_n\subseteq\mathcal{V}$, this integral can be well approximated using standard quadrature rules such as the Gauss-Legendre or Tanh-sinh quadrature.}
\begin{equation}\label{e6}
\hat\theta_n\in\arg\max_{\theta\in\Theta_n}	\frac1n\sum_{i=1}^n\left[\log \int_0^1 f_\varepsilon(y_i-x_i'\beta(u);\sigma)f_c(u|\hat{v}_i;\gamma)du\right].
\end{equation}


\section{Asymptotic Properties}
\label{sec:asymp}

This section establishes consistency and asymptotic normality of the two-step sieve maximum likelihood estimator (2SSMLE), thus permitting nonparametric bootstrap for inference. We distinguish between two cases: (i) the first step is parametric, yielding $\hat\theta_n^{p}$ in Section \ref{subsec:asymp_p}; (ii) the first step is nonparametric, yielding $\hat\theta_n^{np}$ in Section \ref{subsec:asymp_np}. 

\subsection{Parametric First Stage}
\label{subsec:asymp_p}

We begin by imposing some sufficient regularity conditions on the parametric first stage.

\begin{assum}\label{firstasmp_p}
	First Stage Regularity: With i.i.d. $\{X_{1i},Z_i\}_{i=1}^n$,
	\begin{enumerate}[label=(\alph*)]
		\item $\E[\|Z_i\|^4]<\infty$, $\E[X_{1i}^4]<\infty$, and $\E[Z_iZ_i']$ is nonsingular.
		\item CDF $F_\eta$ is continuously differentiable with uniformly bounded continuous density $f_\eta$.
	\end{enumerate} 
\end{assum}

Under the preceding conditions, the generated control variable in (\ref{v_p}) satisfies $|\hat{V}_i-V_i|=O_p(n^{-1/2})$, i.e., it is pointwise $\sqrt{n}$-consistent. Combined with consistency of the sieve approximation in the second stage when $V$ is observed, this yields consistency of $\hat\theta_n^{p}$.

\begin{thm}
	\label{thm2} 
	\textbf{Consistency}\\ 
	Suppose the conditions of Theorem~\ref{thm1} and Assumptions~\ref{parpro}–\ref{firstasmp_p} hold. If the sieve dimension satisfies $J_n \to \infty$ and $J_n / n \to 0$, then the estimator $\hat\theta_n^{p}$ is consistent, that is, $d(\hat\theta_n^{p}, \theta_0) \xrightarrow{p} 0$.
\end{thm}

The proof is given in Appendix~\ref{pf_thm2}. Next, we establish the $L_2$-convergence rate of $\hat\theta_n^{p}$ to the finite-dimensional projection of the true parameter, i.e., $\|\hat\theta_n^{p}-\theta_J^*\|_2^2$. This rate determines how quickly the sieve dimension $J_n$ may grow in the presence of an ill-posed inverse problem. As noted by \citet{chen2006efficient}, the convergence rate of sieve MLE depends on the smoothness level of the likelihood function. Following \cite{fan1991optimal}, we characterize the smoothness of the error distribution through the tail behavior of its characteristic function $\phi_\varepsilon(s)$ as $|s|\to\infty$.\footnote{Although in our model the error distribution is convolved with a copula component, the copula density is uniformly bounded over the trimmed domain $\mathcal{V}_n$, thus does not affect the effective smoothness classification.} Namely, the distribution is \textit{supersmooth} of order $\lambda>0$ if $|\phi_\varepsilon(s)|\asymp |s|^\lambda\exp(-|s|^\lambda/\varsigma)$ as $s\to\infty$, for some constant $\varsigma > 0$. It is \textit{ordinary smooth} of order $\lambda>0$ if $|\phi_\varepsilon(s)|\asymp |s|^{-\lambda}$ as $s\to\infty$. Although the smoothness of the error distribution does not affect consistency of $\hat\theta_n^p$, we require it to be ordinary smooth to establish the asymptotic normality result. While this condition rules out the exactly normal distribution (which is supersmooth of order 2), \citeauthor{hausman2021errors} (\citeyear{hausman2021errors}) demonstrate that even small perturbations from normality typically render the characteristic function ordinary smooth of finite order. Consequently, a sieve based on mixtures of ordinary-smooth densities remains sufficiently flexible to approximate supersmooth distributions while effectively avoiding ill-posedness in finite samples.

We introduce one final assumption, adapted from \cite{hausman2021errors}, which facilitates the derivation of convergence rate for the finite-dimensional distributional parameters.

\begin{assum}\label{cf}
	Variation on characteristic function:\\ Let $\phi_{x\beta,\gamma}(s|x,v)=\int_0^1\exp(isx'\beta(u))f_c(u|v;\gamma)du$, which is the characteristic function of $x'\beta$ conditional on $(x,v)$. There exists a local neighborhood $\mathcal{N} \subseteq \mathcal{X} \times \mathcal{V}$ and a constant $c > 0$ such that, for every $\beta \in M$ and every $s \in [-l,l]$ with $l > 0$ fixed,
	\[
	Var\left(\left|\frac{\phi_{x\beta,\gamma_0}(s|x,v)}{\phi_{x\beta_0,\gamma_0}(s|x,v)}\right|\right)\ge c\cdot \E\left[\left|\frac{\phi_{x\beta,\gamma_0}(s|x,v)-\phi_{x\beta_0,\gamma_0}(s|x,v)}{\phi_{x\beta_0,\gamma_0}(s|x,v)}\right|^2\right],
	\] where the variance and expectation are taken with respect to the joint distribution of $(x,v)$ restricted to $\mathcal{N}$ at the true value $\gamma_0$.
\end{assum} 

Assumption~\ref{cf} requires that, when weighted by the true copula density, there is sufficient variation in $(x,v)$ within the local neighborhood so that the characteristic function is nondegenerate. We now establish the $L_2$-convergence rate of $\|\hat\theta_n^p - \theta_J^*\|_2^2$ as follows.

\begin{lem}\label{lem1}
    \textbf{Convergence Rates}\\
	Suppose the conditions of Theorem~\ref{thm2} and Assumption~\ref{cf} are satisfied. Assume that the error density $f_\varepsilon$ is ordinary smooth of order $\lambda>0$ and that $J_n^{2r+2}/n \to \infty$. Then the estimator $\hat\theta_n^p=(\hat\beta_n(\cdot),\hat\sigma,\hat\gamma)$ obeys the following convergence rates:
	\begin{enumerate}[label=(\alph*)]
		\item $\|\hat\sigma-\sigma_0\|_2^2=O_p\left( \max\left\{\frac{\log n}{n},\frac{\delta\sqrt{-\log \delta}}{\sqrt{n}}\right\}\right)$, where $\delta=\max\{\|\hat\beta_n-\beta_J^*\|_\infty,\|\hat\gamma-\gamma_0\|_2\}$,
		\item $\|\hat\beta_n-\beta_J^*\|_2^2=O_p(J_n^{\frac{2\lambda}{\lambda+1}}n^{-\frac{1}{2(\lambda+1)}})$, and
		\item $\|\hat\gamma-\gamma_0\|_2^2=O_p(n^{-1/2})$, where $\theta_J^*=(\beta_J^*(\cdot),\sigma_0,\gamma_0)$ denotes the projection $\pi_n\theta_0\in\Theta_n$.
	\end{enumerate} 
\end{lem}

The proof is provided in Appendix~\ref{pf_lem1}. The derived convergence rates impose restrictions on how quickly the sieve dimension $J_n$ may diverge, because the second-stage estimation involves an ill-posed inverse problem, under which the minimum eigenvalue of the Hessian can approach zero as the number of basis functions increases. To ensure the validity of the high dimensional central limit theorem, the convergence rates established above must dominate the decay rate of this eigenvalue (See more details in Theorem~\ref{thm3}). This requirement motivates a careful analysis of the rates in Lemma~\ref{lem1}. In particular, combining Lemma~\ref{lem1} with consistency of $\hat\theta_n^p$ yields $\|\hat\sigma-\sigma_0\|_2=o_p(n^{-1/4})$ for the error distribution parameter. By contrast, the convergence rates for the quantile coefficient functions and the copula parameter are slower and depend directly on the first-stage estimation, since the true $V$ is unobserved. As a result, a slow-converging first stage may limit the range of admissible growth rates of $J_n$ required for weak convergence.\footnote{The distinction between $(a)$ and $(b$-$c)$ stems from independence between $\varepsilon$ and $x'\beta$, resulting in the factorization $\phi_{y\mid x,v}(s)=\phi_\varepsilon(s) \cdot \phi_{x'\beta\mid v}(s)$, which is then used to determine the convergence rates.} We now present the asymptotic normality of $\hat\theta_n^p$.

\begin{thm}\label{thm3}
\textbf{Asymptotic Normality}\\
Suppose the conditions of Theorem~\ref{thm2} and Assumption~\ref{cf} are satisfied. Assume that the error density $f_\varepsilon$ is ordinary smooth of order $\lambda>0$. Define the score function $\psi(w;\theta,\pi):=\partial_\theta \log f(y|x,v(x,z;\pi);\theta)$, with associated information matrix $\I(\theta,\pi)=\E[\psi\psi']$ and Hessian $\mathcal{H}(\theta,\pi)=\E[\partial_\theta\psi]$. Assume that there exists a positive sequence $\{\kappa_{J_n}\}$ with $\kappa_{J_n}\to0$ as $J_n\to\infty$ such that the minimum eigenvalue of the information matrix $\I(\theta_0,\pi_0)$ satisfies $\lambda_{\text{min}}(\I(\theta_0,\pi_0))\asymp\kappa_{J_n}$. If $J_n^{\frac{\lambda}{\lambda+1}}n^{\frac{-1}{4(\lambda+1)}}/\kappa_{J_n}\to0$, $J_n^{r+1}\kappa_{J_n}\to\infty$, and $J_n^{r+1}/\sqrt{n}\to\infty$, then 
\begin{align*}
	\sqrt{n \kappa_{J_n}} \, \Omega_{J,\sigma}^{-1/2} (\hat{\sigma} - \sigma_0) \xrightarrow{d} \mathcal{N}(0, \mathbf{I}_{d_\sigma}), \\
\sqrt{n \kappa_{J_n}} \, \Omega_{J,\gamma}^{-1/2} (\hat{\gamma} - \gamma_0) \xrightarrow{d} \mathcal{N}(0, \mathbf{I}_{d_\gamma}),
\end{align*} for each fixed $\tau\in(0,1)$,
\begin{gather*}
	\sqrt{n \kappa_{J_n}} \, \Omega_{J,\tau}^{-1/2} \bigl( \hat{\beta}_n(\tau) - \beta_0(\tau) \bigr) \xrightarrow{d} \mathcal{N}(0, \mathbf{I}_{d_x})\\
	\Omega_{J,\tau}:=(S(\tau)\otimes \mathbf{I}_{d_x})'\:\Omega_{J,\beta}\: (S(\tau)\otimes \mathbf{I}_{d_x}),
\end{gather*} where each $\Omega_{J,\cdot}$ is the corresponding positive definite block (sub-matrix) of
\begin{align*}
	\Omega_J &:= \kappa_{J_n} \mathcal{H}^{-1}(\theta_J^*,\pi_0)\E\big[\phi(W_i;\theta_0,\pi_0)\phi(W_i;\theta_0,\pi_0)'\big]\mathcal{H}^{-1}(\theta_J^*,\pi_0)\\
	\phi(W_i;\theta_0,\pi_0) &:= \psi(W_i;\theta_0,\pi_0) + \E\big[\partial_v\psi(W_i;\theta_0,\pi_0)\cdot \partial_\pi V(X_i,Z_i;\pi_0)\big]\cdot\E[Z_iZ_i']^{-1} Z_i\eta_i,
\end{align*} with the largest eigenvalue of $\Omega_J$ remains bounded. 
\end{thm}

The proof is provided in Appendix~\ref{pf_thm3}. Theorem~\ref{thm3} shows that the 2SSMLE achieves pointwise asymptotic normality provided that the sieve dimension $J_n$ grows at an appropriate rate. In particular, $J_n^{2r+2}/n\to\infty$ requires $J_n$ to grow sufficiently fast so that the sieve approximation bias is adequately controlled. Meanwhile, $J_n^{\frac{\lambda}{\lambda+1}}n^{\frac{-1}{4(\lambda+1)}}/\kappa_{J_n}\to0$ restricts $J_n$ from growing too quickly, thereby preserving weak convergence in the growing dimensional setting and preventing the problem from becoming severely ill-posed. By the information identity, $\mathcal{I}(\theta_0,v) = -\mathcal{H}(\theta_0,v)$, so the minimum eigenvalues of $\mathcal{H}(\theta_0,v)$, $\mathcal{H}(\theta_J^*,v)$, and $\mathcal{I}(\theta_0,v)$ share the same decay rate. Consequently, the estimator is $(n\kappa_{J_n})^{-1/2}$ consistent, where $\kappa_{J_n}$ captures the degree of ill-posedness in the inverse problem. Using the same measure of ill-posedness as in \citet{hausman2021errors}, if the problem is \textit{mildly ill-posed}, so that $\kappa_{J_n} \asymp J_n^{-\delta}$, then weak convergence holds under the above conditions if $(1+\delta)\lambda + \delta < \frac{1}{2}(r+1)$. By contrast, if $\kappa_{J_n} \asymp \exp(-J_n^\delta)$, corresponding to the \textit{severely ill-posed} case, then weak convergence fails even though consistency is maintained.

Theorem~\ref{thm3} also indicates that the contribution of the first stage is not simply negligible, even though it attains a parametric rate that dominates the mildly ill-posed second stage. Instead, the generated regressor error propagates through the second step in a first-order way: It enters through a $J_n\times d_z$ loading, which embeds the first-stage randomness into the same high-dimensional sieve score space. More specifically, the asymptotic variance $\Omega_J$ can be decomposed into three components, corresponding to the randomness arising from both stages as well as their interaction. The first-and second-stage randomness are given by  
\begin{gather*}
\Omega_{\text{sec}}:=\kappa_{J_n} \mathcal{H}^{-1}(\theta_J^*,\pi_0)\mathcal{I}(\theta_0,\pi_0)\mathcal{H}^{-1}(\theta_J^*,\pi_0)\\
\Omega_{\text{fir}}:=\kappa_{J_n} \mathcal{H}^{-1}(\theta_J^*,\pi_0)G_\pi(\theta_0,\pi_0)\Sigma_\pi G'_\pi(\theta_0,\pi_0)\kappa_{J_n} \mathcal{H}^{-1}(\theta_J^*,\pi_0),
\end{gather*} where $G_\pi(\theta_0,\pi_0):=\E\big[\partial_v\psi(W_i;\theta_0,\pi_0)\cdot \partial_\pi V(X_i,Z_i;\pi_0)\big]$ and $\Sigma_\pi:=\E[Z_iZ_i']^{-1}\E[\eta_i^2]$. In general, the lack of orthogonality renders $G_\pi\neq0$. However, in the special case without endogeneity, where $f_c(u|v) = f_c(u) = 1$ and $\partial_v \psi = 0$, we have $G_\pi(\theta_0,\pi_0) = 0$. In this case, the variance reduces to that found in \citet{hausman2021errors}, i.e., $\Omega_{\text{sec}}$, and the first-stage contribution disappears. Moreover, by the definition of $\kappa_{J_n}$, the eigenvalues of $\Omega_{\text{sec}}$ are bounded. We show that the full asymptotic variance satisfies $0\le\Omega_{\text{sec}}\le \Omega_J\le (1+C)\Omega_{\text{sec}}$, for some uniform constant $C>0$. Hence, while the first-stage contribution is not asymptotically negligible, the second stage determines both the convergence rate $\sqrt{n\kappa_{J_n}}$ and the shape of the asymptotic variance.

\subsection{Nonparametric First Stage} 
\label{subsec:asymp_np}

We now present the asymptotic theory for the series first-stage estimator defined  in (\ref{v_np}). As emphasised by \citet{imbens2009identification}, the convergence rate of $\hat V$ depends on the smoothness of the conditional distribution $F_{X_1\mid Z}$. We hereby adopt the following assumption.

\begin{assum}
	\label{firstasym_np} First Stage Regularity: \\
	$Z_i\in \R^{d_z}$ has compact support and $F_{X_1|Z}(x_1,z)$ is continuously differentiable of order $d_1$ on the support with derivatives uniformly bounded on the control variable.
\end{assum}

Under Assumptions~\ref{npc} and \ref{firstasym_np}, Lemma 11 of \citet{imbens2009identification} shows that 
\begin{align}
	\E\Big[\sum_{i=1}^n (\hat{V}_i-V_i)^2/n\Big]=O(K_n/n+K_n^{1-2d_1/d_z}).	
\end{align} This rate depends on both the series dimension $K_n$ and the smoothness level of $F_{X_1|Z}$. It consists of a variance term, $(K_n/n)$, and a squared bias term, $(K_n^{1-2d_1/d_z})$.\footnote{The additional factor $K_n$ in the squared bias arises because the predicted values $\hat{V}_i$ are obtained from regressions whose dependent variables vary across observations.} Let $\alpha_n:=(K_n/n+K_n^{1-2d_1/d_z})^{1/2}$. As in the parametric case in Section~\ref{subsec:asymp_p}, when $\alpha_n \to 0$, the estimator $\hat\theta_n^{np}$ is consistent and satisfies the following convergence rates.

\begin{thm}\label{thm4}
	\textbf{Consistency}\\
		Suppose the conditions of Theorem~\ref{thm1} and Assumptions~\ref{parpro}, \ref{sieve}, \ref{firstasym_np} hold. If $\alpha_n\to0$, $J_n \to \infty$, and $J_n / n \to 0$, then the estimator $\hat\theta_n^{np}$ is consistent.
\end{thm}

\begin{lem}\label{lem2}
    \textbf{Convergence Rates}\\
	Suppose the conditions of Theorem~\ref{thm4} and Assumption~\ref{cf} are satisfied. Assume that the error density $f_\varepsilon$ is ordinary smooth of order $\lambda>0$ and that $\alpha_n\to0$, $J_n^{2r+2}/n \to \infty$. Then the estimator $\hat\theta_n^{np}=(\tilde\beta_n(\cdot),\tilde\sigma,\tilde\gamma)$ obeys the following convergence rates:
	\begin{enumerate}[label=(\alph*)]
		\item $\|\tilde\sigma-\sigma_0\|_2^2=O_p\left( \max\left\{\frac{\log n}{n},\frac{\delta\sqrt{-\log \delta}}{\sqrt{n}}\right\}\right)$, where $\delta=\max\{\|\tilde\beta_n-\beta_J^*\|_\infty,\|\tilde\gamma-\gamma_0\|_2\}$,
		\item $\|\tilde\beta_n-\beta_J^*\|_2^2=O_p(J_n^{\frac{2\lambda}{\lambda+1}}{\alpha_n}^{\frac{1}{(\lambda+1)}})$, and
		\item $\|\tilde\gamma-\gamma_0\|_2^2=O_p({\alpha_n})$, where $\theta_J^*=(\beta_J^*(\cdot),\sigma_0,\gamma_0)$ denotes the projection $\pi_n\theta_0\in\Theta_n$.
	\end{enumerate} 
\end{lem}
The proofs of the two results above parallel those of Theorem~\ref{thm2} and Lemma~\ref{lem1}. For the series basis functions, we impose Assumption 2 of \citet{newey1997convergence}. Specifically, let $p^K(z)$ denote the original $K$-dimensional vector of basis functions, and let $P^K(z)$ denote its normalized transformation used in estimation. Assume there exists a nonsingular constant matrix $B$ such that $P^K(z)=Bp^K(z)$ for all $K$; the smallest eigenvalue of $\E[p^K(Z_i)p^K(Z_i)']$ is bounded away from zero uniformly in $K$, and $\sup_{z\in\mathcal{Z}}\|P^K(z)\|\le \zeta_0(K)$ for a sequence $\zeta_0(K)$ satisfying $\zeta_0(K)^2 K/n\to0$ as $n\to\infty$. It is well acknowledged that for power series $\zeta_0(K)\le CK$ and for splines $\zeta_0(K)\le C K^{1/2}$, for some generic positive constant $C$. Denote the conditional distribution function $F_{X_1|Z}$ as $F$, with the true distribution as $F_0$ and the series estimates as $\hat{F}$. Then the pointwise asymptotic distribution of $\hat\theta_n^{np}$ is as follows.

\begin{thm}
	\label{thm5} \textbf{Asymptotic Normality}\\
	Suppose the conditions of Theorem~\ref{thm4} and Assumption~\ref{cf} are satisfied. Assume that the error density $f_\varepsilon$ is ordinary smooth of order $\lambda>0$. Define the score function $\psi(w;\theta,F):=\partial_\theta \log f(y|x,F(x,z);\theta)$, with associated information matrix $\I(\theta,F)=\E[\psi\psi']$ and Hessian $\mathcal{H}(\theta,F)=\E[\partial_\theta\psi]$. Assume that there exists a positive sequence $\{\kappa_{J_n}\}$ with $\kappa_{J_n}\to0$ as $J_n\to\infty$ such that the minimum eigenvalue of the information matrix $\I(\theta_0,F_0)$ satisfies $\lambda_{\text{min}}(\I(\theta_0,F_0))\asymp\kappa_{J_n}$. If $J_n^{\frac{\lambda}{\lambda+1}}{\alpha_n}^{\frac{1}{2(\lambda+1)}}/\kappa_{J_n}\to0$ with $\alpha_n\to0$, $\zeta_0(K_n)^2 K_n/n\to0$, $J_n^{r+1}\kappa_{J_n}\to\infty$, and $J_n^{r+1}/\sqrt{n}\to\infty$, then
	\begin{align*}
		\sqrt{n\kappa_{J_n}}\Omega_{J,\sigma}^{-1/2}(\tilde\sigma-\sigma_0-\mu_\sigma)\xrightarrow{d} \mathcal{N}(0, \mathbf{I}_{d_\sigma}), \\
\sqrt{n \kappa_{J_n}} \, \Omega_{J,\gamma}^{-1/2} (\tilde{\gamma} - \gamma_0-\mu_\gamma) \xrightarrow{d} \mathcal{N}(0, \mathbf{I}_{d_\gamma}),
\end{align*} for each fixed $\tau\in(0,1)$,
\begin{gather*}
	\sqrt{n \kappa_{J_n}} \, \Omega_{J,\tau}^{-1/2} \bigl( \hat{\beta}_n(\tau) - \beta_0(\tau)-\mu_\beta(\tau) \bigr) \xrightarrow{d} \mathcal{N}(0, \mathbf{I}_{d_x})\\
	\Omega_{J,\tau}:=(S(\tau)\otimes \mathbf{I}_{d_x})'\:\Omega_{J,\beta}\: (S(\tau)\otimes \mathbf{I}_{d_x}),
\end{gather*} where each $\Omega_{J,(\cdot)}$ is the corresponding positive definite block (submatrix) of
\begin{gather*}
	\Omega_J :=\kappa_{J_n}\H^{-1}(\theta_J^*,F_0)\E\big[\phi(W_i,\theta_0)\phi(W_i,\theta_0)'\big]\H^{-1}(\theta_J^*,F_0)\\
	\phi(W_i,\theta_0):=\psi(W_i;\theta_0,F_0)+\phi_{K_n}(W_i,\theta_0,F_0)\\
	\phi_{K_n}(W_i,\theta_0):=\E\Big[\partial_v \psi(W;\theta_0,F_0)\big( \1\{X_{1i}\le X_1\}-F_0(X_{1}|Z_i)\big)P'(Z)\Big]\E[P_iP_i']^{-}P_i,
\end{gather*} with $P_i=P^{K_n}(Z_i)$ and the largest eigenvalue of $\Omega_J$ remains bounded. Bias vectors $\mu_{(\cdot)}$ are the corresponding subvectors of 
	\begin{align*}
		\mu:=-\H^{-1}(\theta_J^*,F_0)\E[\partial_v\psi(W_i;\theta_0,F_0)b(v_i)],\quad b(v_i)=O_p(K_n^{1/2-d_1/d_z}).
	\end{align*} Moreover, if $\sqrt{\frac{n}{\kappa_{J_n}}}K_n^{1/2-d_1/d_z}\to0$, such bias is undersmoothed to be negligible.
\end{thm}

See Appendix~\ref{pf_thm5} for the proof. The above theorem extends the asymptotic normality result to the case of a nonparametric first stage. We show that the first-stage error can be governed so that it does not break down the limiting distribution, while it introduces an additional regularization bias term arising from series approximation error, which enters linearly through the score derivative $\partial_v\psi$. The resulting influence function includes a U-statistic projection that accounts for the dependence between the nonparametric control residuals and the second-stage score. As in Theorem~\ref{thm3}, the influence of the first stage on both the asymptotic variance and the bias term disappears only in the absence of endogeneity, where $\partial_v\psi=0$. In practice, the above result guarantees flexible adoption of nonparametric control-function estimators without sacrificing valid inference, provided the regularity and undersmoothing conditions are satisfied. The empirical selection of the optimal series dimension $K_n$ and sieve dimension $J_n$ is beyond the scope of this paper.

\subsection{Inference via Bootstrap}
\label{subsec:inference}

Although Theorem~\ref{thm3} and \ref{thm5} establishes asymptotic normality of the two-stage sieve MLE, direct variance estimation can be cumbersome depending on the distributional specification. In practice, we therefore recommend inference using nonparametric pairs bootstrap, following the same principle as in \citet{chen2003estimation}. Specifically, for each bootstrap replication $b=1,\dots,B$, draw a bootstrap sample $\{(Y_i^{*(b)},X_i^{*(b)},Z_i^{*(b)})\}_{i=1}^n$ from the empirical distribution of the data. We re-estimate the full two-step procedure to obtain bootstrap control variables $\hat V_i^{*(b)}$ and the second-stage sieve MLE $\hat\theta_n^{*(b)}$. We then construct pointwise $(1-\alpha)$ confidence intervals using the bootstrap standard errors: $\hat\beta_k(\tau)\pm z_{1-\alpha/2}\,\widehat{\mathrm{se}}_{k}(\tau)$, where $z_{1-\frac{\alpha}{2}}$ denotes the $(1-\frac{\alpha}{2})$ standard normal quantile and
\[
\widehat{\mathrm{se}}_{k}(\tau)
=
\left[
\frac{1}{B-1}
\sum_{b=1}^B
\left(
\hat\beta_k^{*(b)}(\tau)
-
\bar\beta_k^*(\tau)
\right)^2
\right]^{1/2},
\qquad
\bar\beta_k^*(\tau)
=
\frac{1}{B}
\sum_{b=1}^B
\hat\beta_k^{*(b)}(\tau),
\] for each quantile index $\tau$ of interest. Analogous intervals can be constructed for finite-dimensional parameters. Such procedure automatically accounts for estimation uncertainty in both stages. We establish its validity in Appendix~\ref{subsec:bootstrap}.

\section{Simulation}
\label{sec:sim}

In this section, we examine the finite-sample performance of the proposed estimator using Monte Carlo simulations. The data-generating-process (DGP) is given by
\begin{align*}
	Y=\beta_0(U)+X_{1}\beta_1(U)+X_2\beta_2(U)+\varepsilon.
\end{align*} The latent rank is generated as $U=\Phi(A)$, where $A\sim\mathcal{N}(0,1)$ and $\Phi$ denotes the standard normal CDF, so that $U\sim\mathcal{U}[0,1]$. The coefficient functions are specified as
\begin{align*}
	\begin{pmatrix}
		\beta_0(u)\\ \beta_1(u)\\ \beta_2(u)
	\end{pmatrix}=
	\begin{pmatrix}
		1+3u-u^2\\ \exp(u) \\ \sqrt{u}
	\end{pmatrix}.
\end{align*} The model includes one endogenous regressor $X_1$ and one exogenous control $X_2$, such that $X_1\ind X_2$ and $X_2\sim \text{LN}(0,1)$. We generate $X_1\sim\text{LN}(0,1)$ by:
\begin{align*}
	\log X_1=\sqrt{\delta^2} Z+\sqrt{1-\delta^2} \eta,
\end{align*} where $Z,\eta\sim \mathcal{N}(0,1)$ and $Z\ind \eta$. The parameter $\delta\in (0,1)$ controls instrument strength and is set to $0.5$ in our baseline DGP. In the baseline scenario, endogeneity is introduced through a Gaussian copula linking the two latent variables $U$ and $V$. Specifically,
\begin{align*}
	V = \Phi(\eta) = \Phi(\rho A+\sqrt{1-\rho^2} B),
\end{align*}
where $B\sim \mathcal{N}(0,1)$, $B\ind A$, and $\rho\in(-1,1)$ governs the degree of endogeneity. Note that $X_1$ is strictly increasing in $\eta$, and $V$ is the CDF of $\eta$. Finally, we draw mean-zero measurement error $\varepsilon$ from a mixed Gaussian distribution following \citet{hausman2021errors}:
\begin{align*}
	\varepsilon\sim \left\{ \begin{array}{lcl}
\mathcal{N}(-3,1), &\text{w.p }0.5,\\
\mathcal{N}(2,1), &\text{w.p }0.25,\\
\mathcal{N}(4,1), &\text{w.p }0.25.
\end{array}\right.
\end{align*}  

We illustrate the two-step sieve maximum likelihood estimator (2SSMLE) using both parametric and nonparametric first-stage estimation. In the parametric case, assume the true specification is fully known, so that $V$ can be estimated by directly plugging in OLS estimates. More generally, in the nonparametric case, we run a spline regression of $X_1$ on $(Z,X_2)$ and then construct $\hat{V}_i=\hat{F}(X_{1i}|Z_i,X_{2i})$ as in (\ref{v_np}). The second step plugs in first-step estimates $\hat{v}$ and solves the empirical log-likelihood maximization problem specified in (\ref{e6}). We use B-splines of order $r$ with $J_n$ grid knots to approximate the unknown coefficient functions at $\{\beta_k(\tau_j)\}_{j=1}^{J_n}$. Given the properties of the measurement error and copula distributions, the integral within the log-likelihood can be well approximated by Gauss-Legendre quadrature with $q$ nodes per interval on the same partition used for the B-spline basis.

Despite the sophistication of the second-stage likelihood and optimization problem, the proposed estimator remains computationally efficient under our implementation. Specifically, we solve the optimization problem using a gradient-based constrained interior-point method, with analytical gradients supplied to stabilize the descent path and facilitate convergence. Stochastic gradient descent may also be employed to mitigate the risk of convergence to local stationary points. Additional implementation details are provided in Appendix~\ref{subsec:algorithm}.

Our 2SSMLE estimator simultaneously corrects the bias from endogeneity and measurement error. In contrast, standard quantile regression (QR) is biased from both sources. The control-function QR (hereafter denoted CFQR), such as \citep{lee2007endogeneity}, addresses only endogeneity, and the sieve MLE (hereafter denoted SMLE) of \citep{hausman2021errors} addresses only measurement error but ignores possible endogeneity. For CFQR in our simulations, we implement an adaptation of \cite{lee2007endogeneity}: the control variable $\hat V$ is estimated from our first stage, and the second runs a partial linear QR of $Y$ on $(X_1,X_2,\hat{V})$.\footnote{Specifically, define $P_k(w)=[x_1,x_2,p_1(v),\dots,p_k(v)]'$ for a power-series basis $\{p_k:k=1,2,\dots\}$. With a trimming function $t(w)=\1\{w\in\mathcal{W}\}$ to limit unduly values \citep{lee2007endogeneity,blundell2007censored}, CFQR solves $\min_\theta S_{nk}(\theta)=\frac1n\sum_{i=1}^n t(\hat{W}_i)\rho_\tau[Y_i-P_k(\hat{W}_i)'\theta]$ for any $\tau\in(0,1)$ and $\hat{W}_i=(X_{1i},X_{2i},\hat{V}_i)'$.} 

We first illustrate how standard estimators are biased across quantiles when both endogeneity and measurement error are present. We then shut down one source of bias at a time to assess the cost of accounting for features that may be absent. Finally, we examine the robustness of 2SSMLE to misspecification and evaluate bootstrap validity.

\begin{figure}[ht]
	\centering
	\vspace{1em}
	\caption{MC Results Comparison: $\varepsilon\sim 3\mathcal{N}$ and $\rho=0.5$} 
	\includegraphics[scale=0.5]{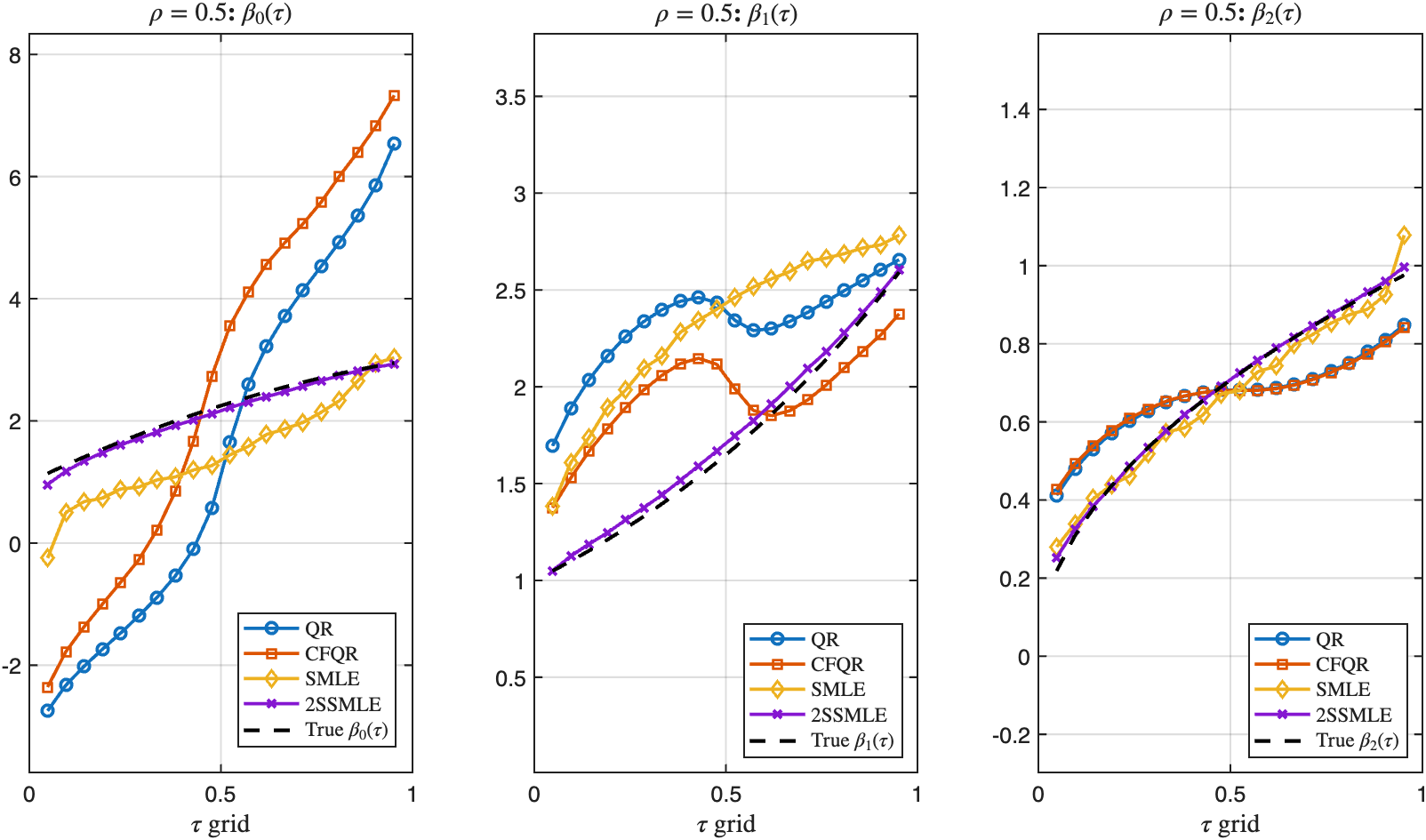}
    \caption*{{\footnotesize Note: Figure~\ref{fig1} compares the true coefficient functions with four estimators—QR (circle), QR with control function (square), sieve MLE (diamond), and two-step sieve MLE with control function (cross)—based on 500 Monte Carlo replications with sample size $n=5000$. The DGP follows Section~\ref{sec:sim}, with measurement error drawn from a three-component Gaussian mixture and copula dependence $\rho=0.5$. CFQR implements an adaptation of \citet{lee2007endogeneity} with $k=5$ regression splines (to the fourth order). SMLE replicates \citep{hausman2021errors}. 2SSMLE uses 20 grid knots, B-splines at $r=1$, and $5$ knots per interval for the quadrature.}}
    \label{fig1}
\end{figure}

Figure~\ref{fig1} plots the average point estimates $\hat\beta(\tau_j)$ at 20 grid knots to visually compare the bias patterns across methods. We first focus on the slope coefficient function $\beta_1(\cdot)$ (Figure~\ref{fig1}, middle panel), while the intercept function $\beta_0(\cdot)$ shown in the left panel exhibits analogous patterns. CFQR (red squares) is biased due to the measurement error $\varepsilon$: lower quantiles biased upwards and upper quantiles downwards, reflecting the influence of the two spikes of the $\varepsilon$-distribution at $-3$ and $3$, as shown in Figure~\ref{fig:error}. SMLE (yellow diamonds), unaffected by EIV, tracks the shape of the true $\beta_1$ but is upward biased due to endogeneity. Simple QR (blue circles) suffers from both issues. By comparison, 2SSMLE (purple crosses) exhibits much smaller bias at all quantiles: its average absolute bias is about 4\% of the true $\beta_1$, versus 43\% for QR, 24\% for CFQR, and 35\% for SMLE. For the slope coefficient of the exogenous covariate $X_2$ (Figure~\ref{fig1}, right panel), there is no endogeneity issue. Consequently, QR and CFQR closely align with each other and are affected only by measurement error. Likewise, SMLE and 2SSMLE closely overlap--both track the true coefficient function well.

\begin{figure}[ht]
	\centering
	\vspace{1em}
	\caption{MC Results: 2SSMLE $\varepsilon\sim 3\mathcal{N}$} 
	\includegraphics[scale=0.5]{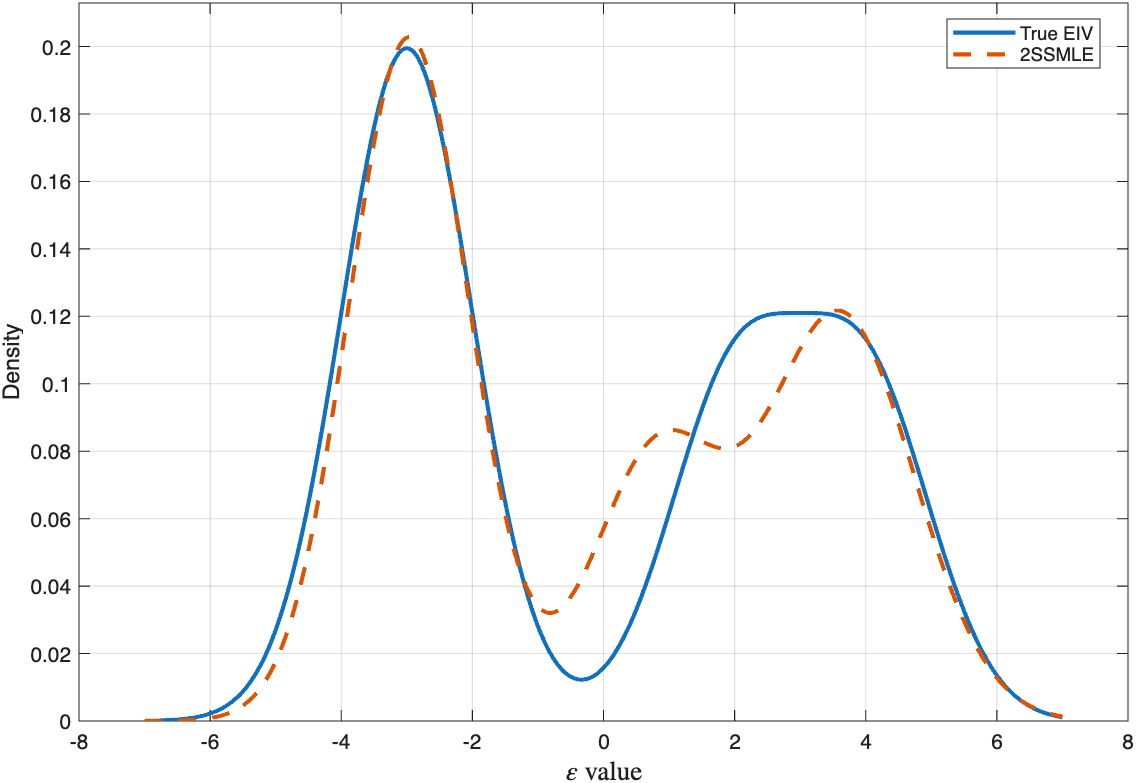}
	\caption*{{\footnotesize Note: Figure~\ref{fig:error} plots the true measurement error density (solid line) against average 2SSMLE estimated density (dash line) from 500 Monte Carlo simulations with sample size $n=5000$. The DGP follows Section~\ref{sec:sim}, with measurement error drawn from a three-component Gaussian mixture and copula dependence $\rho=0.5$. 2SSMLE uses 20 grid knots, B-splines at $r=1$, and $5$ knots per interval for the quadrature.}}
	\label{fig:error}
\end{figure}

Such bias worsens as endogeneity exacerbates: In Figure~\ref{fig2}, both QR and SMLE results in larger deviation as the correlation between $U$ and $V$ increases from $\rho=0.5$ to $\rho=0.9$. Note that the impact of EIV depends on its magnitude relative to the true outcome. In Figure~\ref{fig3}, when $\varepsilon\sim\mathcal N(0,1)$, the EIV-induced distortion is small, making SMLE and QR almost indistinguishable, both primarily reflecting endogeneity rather than measurement error. Finally, we compare 2SSMLE with the default estimators in settings where either only endogeneity or measurement error is absent, thereby illustrating the cost of conservatively accounting for features that may not actually be present. As shown in Table~\ref{tab:loss}, 2SSMLE generally outperforms SMLE when there is no endogeneity and CFQR when there is no EIV, although it produces slightly larger bias in the slope coefficient function than CFQR due to misspecification in the likelihood model.\footnote{Although our identification framework does not formally allow $\varepsilon\stackrel{a.s.}{=}0$, the 2SSMLE can approximate a kernel-smoothed indicator function through the error distribution, making it applicable to multivalued discrete measurement error as well.}

In the absence of knowledge about the true distributional families, we illustrate the flexibility of the Gaussian–mixture specification by showing that it can approximate a variety of alternative EIV distributions following \citet{hausman2021errors}, including student-$t$ and Laplace. Numerical details of the aforementioned results can be found in Tables~\ref{tab:gaumixrho1} to \ref{tab:semiepsilon}. Table~\ref{tab:pfirstboot} shows the bootstrap standard errors and average coverage of our bootstrap confidence interval. At a significance level $\alpha=0.05$, the overall coverage is close to $95\%$. Without further specification, all the above implementation of 2SSMLE uses a parametric first stage. Finally, we repeat the above exercises using series-estimated control functions. Table~\ref{tab:npfirstboot} shows that results parallel the pattern seen in the parametric framework. More details on mean bias and MSE of different methods, bootstrap standard errors, bootstrap coverage and additional simulation results are provided in Appendix~\ref{subsec:simresults}.

\section{Conclusion}
\label{sec:con}

This paper studies quantile regression in the presence of an endogenous regressor and additive measurement error in the dependent variable. After isolating the endogenous treatment from the latent rank with a control variable, we establish nonparametric identification of the conditional quantile coefficient function. Assuming the nuisance distributions are known up to some finite-dimensional parameters, we propose a plug-in two-step sieve maximum likelihood estimator (2SSMLE), which combines a control-function first stage with a sieve ML second stage, incorporating the generated control variable through copula weights. The first-stage control function can be estimated flexibly by series regression. When the series dimension and the sieve dimension grow at appropriate rates, the estimator is consistent and asymptotically normal, with convergence rates governed by the degree of ill-posedness. Monte Carlo simulations show that our estimator substantially outperform existing methods in settings where both endogeneity and measurement error may be present. In future work, a substantive empirical application—such as revisiting Engel curve estimation in \citet{blundell2007semi}—would provide an opportunity to examine how outcome mismeasurement interacts with endogeneity in practice, and to assess the empirical consequences of our proposed corrections for distributional treatment effects.


\bibliography{eeq_references.bib}

\pagebreak
\section{Appendix}
\label{sec:appendix}

\subsection{Appendix A1: 2SSMLE Implementation}
\label{subsec:algorithm}

We choose the starting values of each parameter as follows.
\begin{enumerate}[label=(\alph*)]
	\item Quantile coefficient functions $\beta$: \\Run QR at $\tilde{J}_n$ evaluation points to obtain $\{\hat\beta^{QR}_k(\tau_j)\}_{j=1}^{\tilde{J}_n}$. We then plug in $\hat\beta^{QR}$ and implement an Expectation-Maximization (EM) algorithm for up to 50 iterations. Similar to \citet{hausman2021errors}, the M-step reduces to a weighted least squares (WLS) problem under Gaussian mixture errors and a Gaussian copula specification, where the weights are observation-specific posterior probabilities determined by the current parameter values and the copula structure. Finally, we project the resulting $\hat\beta^{WLS}$ onto the $J_n$-dimensional B-spline basis via least squares to obtain the initial value $\hat\beta^{\text{start}}$.
	\item Distributional parameters $(\sigma,\gamma)$:\\ The copula correlation parameter $\hat\rho^{\text{start}}$ is initialized as the sample correlation between the residuals from regressing $\log(X_1)$ on $(Z,X_2)$ and the residuals from regressing $Y$ on $(X_1,X_2)$. For the Gaussian mixture error distribution, we initialize the component weights equally at $1/3$, set the component means symmetrically around zero with magnitude $\mathrm{SD}(y)/\sqrt{2}$,, and initialize the component standard deviations at $\mathrm{SD}(y)/\sqrt{2}$.
\end{enumerate}

In practice, where the true distributions are unknown, we recommend increasing the number of Gaussian mixture components for the measurement error and incorporating mixtures of heterogeneous copulas following \citet{burda2008bayesian,burda2012poisson,khaled2017approximating,qu2021copula}.\footnote{\citet{khaled2017approximating} show that finite mixtures of elliptical copulas (e.g. Gaussian) or Archimedean copulas are not dense in the space of all copulas. \citet{qu2021copula} thus suggest that a finite mixture of heterogeneous parametric copulas is flexible enough to capture more variations in the dependence structures.} In our simulation, we focus on a single Gaussian copula. 

Because the covariates are strictly positive, as in many empirical applications and related theoretical work \citep{hausman2021errors,chernozhukov2009improving}, we impose monotonicity on each component $\hat\beta_{nk}$. This condition is stronger than the theoretical requirement that $x'\beta(\cdot)$ be monotonically increasing for all $x\in\mathcal{X}$, since it rules out local non-monotonic perturbations in each component function. Nevertheless, such a monotonicity constraint is crucial for numerical stability.\footnote{Note that monotonicity cannot be imposed by rearranging the estimated functions $x'\hat\beta_{nk}$, as suggested by \citet{chernozhukov2009improving}, because the likelihood integrand $g(\beta(u),u):=f_\varepsilon(y-x'\beta(u))f_c(u\mid v)$ depends on $u$ both through $\beta(\cdot)$ and through the copula weighting function.} 

\vspace*{-0.5em}
\subsection{Appendix A2: Simulation Results}
\label{subsec:simresults}

\begin{landscape}
\begin{table}[htbp]
  \centering
  \caption{MC Results Comparison: $\varepsilon\sim 3\mathcal{N}$ and $\rho=0.5$}
  \begin{adjustbox}{width=0.85\linewidth}
    \begin{tabular}{c|cccc|cccc|cccc}
    \toprule
    \toprule
    \multicolumn{1}{c}{} & \multicolumn{12}{c}{\textit{I. Mean Bias}} \\
    \midrule
          & \multicolumn{4}{c|}{$\beta_0$}    & \multicolumn{4}{c|}{$\beta_1$}    & \multicolumn{4}{c}{$\beta_2$} \\
    \midrule
    Quantile & QR    & CFQR  & SMLE  & 2SSMLE & QR    & CFQR  & SMLE  & 2SSMLE & QR    & CFQR  & SMLE  & 2SSMLE \\
    \midrule
    0.1 & -4.381 & -3.695 & -0.777 & -0.101 & 0.925 & 0.520 & 0.510 & 0.026 & 0.172 & 0.184 & 0.031 & 0.017 \\
    0.2 & -3.572 & -2.629 & -0.800 & -0.055 & 1.009 & 0.591 & 0.683 & 0.036 & 0.134 & 0.142 & 0.002 & -0.003 \\
    0.3 & -2.891 & -1.802 & -0.860 & -0.062 & 1.005 & 0.588 & 0.765 & 0.046 & 0.094 & 0.099 & -0.017 & 0.000 \\
    0.4 & -2.156 & -0.859 & -0.915 & -0.076 & 0.935 & 0.520 & 0.818 & 0.053 & 0.049 & 0.050 & -0.032 & 0.002 \\
    0.5 & -1.277 & 0.186 & -0.934 & -0.088 & 0.815 & 0.393 & 0.790 & 0.059 & -0.008 & -0.010 & -0.018 & 0.001 \\
    0.6 & -0.262 & 1.218 & -0.806 & -0.076 & 0.666 & 0.222 & 0.744 & 0.054 & -0.073 & -0.075 & -0.027 & 0.002 \\
    0.7 & 0.745 & 2.035 & -0.687 & -0.078 & 0.518 & 0.052 & 0.647 & 0.053 & -0.120 & -0.122 & -0.018 & 0.001 \\
    0.8 & 1.704 & 2.820 & -0.564 & -0.048 & 0.378 & -0.079 & 0.522 & 0.040 & -0.143 & -0.146 & -0.020 & 0.003 \\
    0.9 & 2.752 & 3.853 & -0.182 & -0.020 & 0.234 & -0.180 & 0.360 & 0.027 & -0.146 & -0.152 & -0.037 & 0.006 \\
    \midrule
          &       &       &       & \multicolumn{1}{c}{} &       &       &       & \multicolumn{1}{c}{} &       &       &       &  \\
          & \multicolumn{12}{c}{\textit{II. Mean Squared Error}} \\
    \midrule
          & \multicolumn{4}{c|}{$\beta_0$}    & \multicolumn{4}{c|}{$\beta_1$}    & \multicolumn{4}{c}{$\beta_2$} \\
    \midrule
    Quantile & QR    & CFQR  & SMLE  & 2SSMLE & QR    & CFQR  & SMLE  & 2SSMLE & QR    & CFQR  & SMLE  & 2SSMLE \\
    \midrule
    0.1 & 19.216 & 16.591 & 1.030 & 0.057 & 0.863 & 0.285 & 0.354 & 0.012 & 0.033 & 0.037 & 0.057 & 0.004 \\
    0.2 & 12.774 & 9.912 & 1.111 & 0.038 & 1.023 & 0.361 & 0.559 & 0.013 & 0.020 & 0.023 & 0.055 & 0.003 \\
    0.3 & 8.377 & 6.437 & 1.212 & 0.035 & 1.013 & 0.356 & 0.670 & 0.015 & 0.011 & 0.012 & 0.057 & 0.003 \\
    0.4 & 4.664 & 4.334 & 1.304 & 0.034 & 0.878 & 0.282 & 0.741 & 0.016 & 0.005 & 0.005 & 0.059 & 0.003 \\
    0.5 & 1.649 & 4.080 & 1.268 & 0.036 & 0.668 & 0.167 & 0.682 & 0.015 & 0.003 & 0.003 & 0.063 & 0.003 \\
    0.6 & 0.086 & 5.052 & 1.038 & 0.032 & 0.447 & 0.062 & 0.598 & 0.014 & 0.008 & 0.008 & 0.061 & 0.003 \\
    0.7 & 0.570 & 7.515 & 0.768 & 0.031 & 0.271 & 0.016 & 0.451 & 0.012 & 0.017 & 0.018 & 0.047 & 0.003 \\
    0.8 & 2.917 & 11.026 & 0.533 & 0.027 & 0.146 & 0.019 & 0.293 & 0.009 & 0.023 & 0.024 & 0.039 & 0.002 \\
    0.9 & 7.586 & 18.018 & 0.195 & 0.018 & 0.058 & 0.045 & 0.141 & 0.006 & 0.024 & 0.026 & 0.031 & 0.002 \\
    \bottomrule
    \bottomrule
    \end{tabular}%
    \end{adjustbox}
  \caption*{{\footnotesize Notes: Table~\ref{tab:gaumixrho1} reports mean bias and MSE for four estimators across 500 Monte Carlo replications with sample size $n=5,000$. The DGP follows Section~\ref{sec:sim}, with $\varepsilon$ drawn from a three-component Gaussian mixture and copula dependence $\rho=0.5$. CFQR implements an adaptation of \citet{lee2007endogeneity} with $k=5$ regression splines (to the fourth order). SMLE replicates  \citep{hausman2021errors}. 2SSMLE uses 20 grid knots, B-splines at $r=1$, and $5$ knots per interval for the quadrature. The upper panel reports mean bias; the lower panel reports MSE.}}%
  \label{tab:gaumixrho1}%
\end{table}%
\end{landscape}

\pagebreak

\begin{landscape}
\begin{table}[htbp]
  \centering
  \caption{MC Results Comparison: $\varepsilon\sim 3\mathcal{N}$ and $\rho=0.9$}
  \begin{adjustbox}{width=0.85\linewidth}
    \begin{tabular}{c|cccc|cccc|cccc}
    \toprule
    \toprule
    \multicolumn{1}{c}{} & \multicolumn{12}{c}{\textit{I. Mean Bias}} \\
    \midrule
          & \multicolumn{4}{c|}{$\beta_0$}    & \multicolumn{4}{c|}{$\beta_1$}    & \multicolumn{4}{c}{$\beta_2$} \\
    \midrule
    Quantile & QR    & CFQR  & SMLE  & 2SSMLE & QR    & CFQR  & SMLE  & 2SSMLE & QR    & CFQR  & SMLE  & 2SSMLE \\
    \midrule
    0.1   & -5.058 & -3.883 & -1.265 & 0.018 & 1.672 & 0.973 & 1.293 & -0.039 & 0.222 & 0.258 & 0.086 & 0.006 \\
    0.2   & -4.170 & -2.662 & -1.328 & -0.011 & 1.602 & 0.934 & 1.421 & 0.015 & 0.167 & 0.185 & 0.046 & 0.001 \\
    0.3   & -3.444 & -1.761 & -1.340 & -0.034 & 1.501 & 0.856 & 1.415 & 0.018 & 0.114 & 0.118 & 0.016 & 0.001 \\
    0.4   & -2.662 & -0.823 & -1.317 & -0.041 & 1.373 & 0.743 & 1.353 & 0.014 & 0.058 & 0.057 & -0.021 & 0.004 \\
    0.5   & -1.714 & 0.246 & -1.245 & -0.017 & 1.220 & 0.590 & 1.241 & 0.010 & -0.010 & -0.014 & -0.066 & 0.001 \\
    0.6   & -0.623 & 1.244 & -1.171 & -0.022 & 1.052 & 0.404 & 1.105 & 0.009 & -0.084 & -0.089 & -0.066 & 0.000 \\
    0.7   & 0.407 & 2.195 & -1.079 & -0.008 & 0.884 & 0.228 & 0.949 & 0.001 & -0.140 & -0.144 & -0.084 & 0.001 \\
    0.8   & 1.361 & 3.173 & -0.951 & -0.017 & 0.706 & 0.062 & 0.764 & 0.006 & -0.167 & -0.179 & -0.085 & -0.001 \\
    0.9   & 2.393 & 4.351 & -0.694 & -0.033 & 0.510 & -0.110 & 0.557 & 0.009 & -0.172 & -0.196 & -0.024 & -0.001 \\
    \midrule
          &       &       &       & \multicolumn{1}{c}{} &       &       &       & \multicolumn{1}{c}{} &       &       &       &  \\
          & \multicolumn{12}{c}{\textit{II. Mean Squared Error}} \\
    \midrule
          & \multicolumn{4}{c|}{$\beta_0$}    & \multicolumn{4}{c|}{$\beta_1$}    & \multicolumn{4}{c}{$\beta_2$} \\
    \midrule
    Quantile & QR    & CFQR  & SMLE  & 2SSMLE & QR    & CFQR  & SMLE  & 2SSMLE & QR    & CFQR  & SMLE  & 2SSMLE \\
    \midrule
    0.1   & 20.948 & 17.765 & 1.896 & 0.032 & 2.797 & 0.956 & 1.745 & 0.054 & 0.052 & 0.069 & 0.058 & 0.004 \\
    0.2   & 14.471 & 9.334 & 2.015 & 0.033 & 2.569 & 0.880 & 2.053 & 0.029 & 0.030 & 0.036 & 0.059 & 0.003 \\
    0.3   & 9.416 & 5.562 & 2.058 & 0.039 & 2.253 & 0.742 & 2.022 & 0.021 & 0.015 & 0.016 & 0.064 & 0.002 \\
    0.4   & 4.912 & 3.797 & 1.957 & 0.042 & 1.887 & 0.563 & 1.837 & 0.017 & 0.005 & 0.005 & 0.067 & 0.002 \\
    0.5   & 1.392 & 3.952 & 1.778 & 0.041 & 1.490 & 0.361 & 1.546 & 0.015 & 0.002 & 0.002 & 0.079 & 0.002 \\
    0.6   & 0.022 & 5.519 & 1.579 & 0.034 & 1.109 & 0.176 & 1.228 & 0.012 & 0.009 & 0.010 & 0.085 & 0.002 \\
    0.7   & 0.804 & 8.333 & 1.316 & 0.030 & 0.784 & 0.063 & 0.910 & 0.009 & 0.022 & 0.023 & 0.082 & 0.002 \\
    0.8   & 3.439 & 13.519 & 1.034 & 0.027 & 0.500 & 0.013 & 0.592 & 0.007 & 0.030 & 0.034 & 0.081 & 0.002 \\
    0.9   & 9.314 & 22.230 & 0.583 & 0.025 & 0.261 & 0.021 & 0.322 & 0.005 & 0.032 & 0.041 & 0.075 & 0.001 \\
    \bottomrule
    \bottomrule
    \end{tabular}%
  \end{adjustbox}
  \caption*{{\footnotesize Notes: Table~\ref{tab:gaumixrho2} reports mean bias and MSE for four estimators across 500 Monte Carlo replications with sample size $n=5,000$. The DGP follows Section~\ref{sec:sim}, with $\varepsilon$ drawn from a three-component Gaussian mixture and copula dependence $\rho=0.9$. CFQR implements an adaptation of \citet{lee2007endogeneity} with $k=5$ regression splines (to the fourth order). SMLE replicates  \citep{hausman2021errors}. 2SSMLE uses 20 grid knots, B-splines at $r=1$, and $5$ knots per interval for the quadrature. The upper panel reports mean bias; the lower panel reports MSE.}}%
  \label{tab:gaumixrho2}%
\end{table}%
\end{landscape}

\pagebreak
\begin{landscape}
\vspace*{\fill}
\begin{table}[htbp]
  \centering
  \caption{2SSMLE: Robustness to $F_\varepsilon$ Misspecification}
  \begin{adjustbox}{width=0.85\linewidth}
    \begin{tabular}{c|cc|cc|cc|cc|cc|cc}
    \toprule
    \toprule
          & \multicolumn{6}{c|}{student-t}                & \multicolumn{6}{c}{Laplace} \\
\cmidrule{2-13}          & \multicolumn{2}{c|}{$\beta_0$} & \multicolumn{2}{c|}{$\beta_1$} & \multicolumn{2}{c|}{$\beta_2$} & \multicolumn{2}{c|}{$\beta_0$} & \multicolumn{2}{c|}{$\beta_1$} & \multicolumn{2}{c}{$\beta_2$} \\
    \midrule
    Quantile & mean bias & MSE   & mean bias & MSE   & mean bias & MSE   & mean bias & MSE   & mean bias & MSE   & mean bias & MSE \\
    \midrule
    0.1   & 0.100 & 0.193 & -0.036 & 0.037 & -0.003 & 0.008 & -0.232 & 0.444 & -0.051 & 0.043 & -0.002 & 0.007 \\
    0.2   & 0.044 & 0.151 & 0.018 & 0.021 & -0.013 & 0.007 & -0.241 & 0.417 & 0.024 & 0.020 & -0.008 & 0.006 \\
    0.3   & -0.045 & 0.137 & 0.026 & 0.019 & -0.007 & 0.006 & -0.292 & 0.423 & 0.041 & 0.020 & -0.003 & 0.006 \\
    0.4   & -0.096 & 0.145 & 0.026 & 0.019 & -0.004 & 0.006 & -0.336 & 0.446 & 0.042 & 0.020 & 0.001 & 0.005 \\
    0.5   & -0.109 & 0.151 & 0.028 & 0.018 & 0.001 & 0.005 & -0.352 & 0.452 & 0.045 & 0.020 & 0.001 & 0.004 \\
    0.6   & -0.082 & 0.162 & 0.012 & 0.016 & 0.006 & 0.005 & -0.322 & 0.436 & 0.048 & 0.020 & 0.002 & 0.004 \\
    0.7   & -0.060 & 0.155 & 0.019 & 0.016 & 0.004 & 0.004 & -0.309 & 0.442 & 0.047 & 0.015 & 0.000 & 0.005 \\
    0.8   & -0.036 & 0.137 & 0.017 & 0.013 & 0.003 & 0.004 & -0.298 & 0.420 & 0.034 & 0.013 & 0.008 & 0.005 \\
    0.9   & -0.046 & 0.116 & 0.009 & 0.011 & 0.001 & 0.003 & -0.291 & 0.387 & 0.014 & 0.011 & 0.016 & 0.006 \\
    \bottomrule
    \bottomrule
    \end{tabular}%
  \end{adjustbox}
    \label{tab:semiepsilon}%
  \caption*{{\footnotesize Notes: Table~\ref{tab:semiepsilon} reports mean bias and MSE for 2SSMLE that models the measurement error as a three-component Gaussian mixture, based on 500 Monte Carlo replications with sample size $n=5,000$  on 20 grid knots. The DGP follows Section~\ref{sec:sim}, except that $\varepsilon$ is generated from (i) a Student's $t$ distribution with three degrees of freedom, scaled by $\sqrt{3.5}$; or (ii) a Laplace distribution with parameter $\lambda=2.29$. Parameters in (i)-(ii) are selected to ensure the variance of $\varepsilon$ is held constant across simulation designs \citep{hausman2021errors}. 2SSMLE uses 20 grid knots, B-splines at $r=1$, and $5$ knots per interval for the quadrature. }}%
\end{table}%
\vspace*{\fill}
\end{landscape}

\pagebreak
\begin{landscape}
\begin{table}[htbp]
  \centering
  \caption{2SSMLE under Single-source Bias}
  \begin{adjustbox}{width=0.85\linewidth}
    \begin{tabular}{ccc|cc|cc|cc|cc|cc}
    \toprule
    \toprule
          & \multicolumn{12}{c}{\textit{I.Mean Bias}} \\
\cmidrule{2-13}          & \multicolumn{6}{c|}{$\varepsilon\sim3\mathcal{N}$, no endogeneity}      & \multicolumn{6}{c}{no EIV, $\rho=0.5$} \\
\cmidrule{2-13}          & \multicolumn{2}{c|}{$\beta_0$} & \multicolumn{2}{c|}{$\beta_1$} & \multicolumn{2}{c|}{$\beta_2$} & \multicolumn{2}{c|}{$\beta_0$} & \multicolumn{2}{c|}{$\beta_1$} & \multicolumn{2}{c}{$\beta_2$} \\
    \midrule
    Quantile & SMLE  & 2SSMLE & SMLE  & 2SSMLE & SMLE  & 2SSMLE & CFQR  & 2SSMLE & CFQR  & 2SSMLE & CFQR  & 2SSMLE \\
    \midrule
    0.1   & -0.298 & -0.182 & 0.035 & 0.049 & 0.043 & 0.011 & -0.378 & -0.578 & 0.178 & -0.013 & -0.003 & -0.019 \\
    0.2   & -0.134 & -0.130 & 0.034 & 0.057 & 0.010 & 0.000 & -0.514 & -0.588 & 0.246 & -0.013 & -0.013 & -0.012 \\
    0.3   & -0.119 & -0.134 & 0.026 & 0.067 & 0.046 & 0.001 & -0.566 & -0.605 & 0.289 & -0.009 & -0.021 & -0.001 \\
    0.4   & -0.161 & -0.142 & 0.037 & 0.076 & 0.065 & 0.006 & -0.572 & -0.623 & 0.313 & -0.002 & -0.027 & 0.010 \\
    0.5   & -0.188 & -0.148 & 0.058 & 0.080 & 0.057 & 0.005 & -0.512 & -0.644 & 0.322 & 0.007 & -0.030 & 0.022 \\
    0.6   & -0.218 & -0.152 & 0.094 & 0.086 & 0.053 & 0.007 & -0.429 & -0.666 & 0.309 & 0.018 & -0.032 & 0.032 \\
    0.7   & -0.244 & -0.135 & 0.103 & 0.080 & 0.050 & 0.004 & -0.269 & -0.681 & 0.277 & 0.027 & -0.032 & 0.039 \\
    0.8   & -0.034 & -0.065 & 0.017 & 0.054 & 0.048 & 0.004 & -0.090 & -0.687 & 0.226 & 0.036 & -0.029 & 0.040 \\
    0.9   & 0.112 & 0.003 & 0.062 & 0.030 & 0.018 & 0.007 & 0.044 & -0.678 & 0.152 & 0.037 & -0.025 & 0.033 \\
    \midrule
          &       & \multicolumn{1}{c}{} &       & \multicolumn{1}{c}{} &       & \multicolumn{1}{c}{} &       & \multicolumn{1}{c}{} &       & \multicolumn{1}{c}{} &       &  \\
          & \multicolumn{12}{c}{\textit{II.Mean Squared Error}} \\
\cmidrule{2-13}          & \multicolumn{6}{c|}{$\varepsilon\sim3\mathcal{N}$, no endogeneity}      & \multicolumn{6}{c}{no EIV, $\rho=0.5$} \\
\cmidrule{2-13}          & \multicolumn{2}{c|}{$\beta_0$} & \multicolumn{2}{c|}{$\beta_1$} & \multicolumn{2}{c|}{$\beta_2$} & \multicolumn{2}{c|}{$\beta_0$} & \multicolumn{2}{c|}{$\beta_1$} & \multicolumn{2}{c}{$\beta_2$} \\
    \midrule
    Quantile & SMLE  & 2SSMLE & SMLE  & 2SSMLE & SMLE  & 2SSMLE & CFQR  & 2SSMLE & CFQR  & 2SSMLE & CFQR  & 2SSMLE \\
    \midrule
    0.1   & 0.416 & 0.083 & 0.044 & 0.007 & 0.073 & 0.004 & 0.238 & 0.642 & 0.036 & 0.002 & 0.001 & 0.002 \\
    0.2   & 0.383 & 0.058 & 0.063 & 0.010 & 0.081 & 0.004 & 0.377 & 0.636 & 0.067 & 0.003 & 0.001 & 0.002 \\
    0.3   & 0.464 & 0.056 & 0.096 & 0.013 & 0.141 & 0.004 & 0.499 & 0.638 & 0.091 & 0.004 & 0.001 & 0.002 \\
    0.4   & 0.444 & 0.052 & 0.108 & 0.016 & 0.143 & 0.003 & 0.566 & 0.644 & 0.105 & 0.005 & 0.002 & 0.002 \\
    0.5   & 0.375 & 0.053 & 0.100 & 0.017 & 0.116 & 0.003 & 0.609 & 0.660 & 0.109 & 0.005 & 0.002 & 0.002 \\
    0.6   & 0.368 & 0.055 & 0.096 & 0.019 & 0.100 & 0.003 & 0.624 & 0.686 & 0.100 & 0.005 & 0.002 & 0.003 \\
    0.7   & 0.308 & 0.052 & 0.090 & 0.017 & 0.101 & 0.003 & 0.575 & 0.715 & 0.082 & 0.005 & 0.002 & 0.003 \\
    0.8   & 0.149 & 0.038 & 0.065 & 0.011 & 0.110 & 0.002 & 0.487 & 0.739 & 0.055 & 0.004 & 0.001 & 0.003 \\
    0.9   & 0.046 & 0.020 & 0.040 & 0.008 & 0.110 & 0.002 & 0.379 & 0.751 & 0.025 & 0.003 & 0.001 & 0.002 \\
    \bottomrule
    \bottomrule
    \end{tabular}%
    \end{adjustbox}
  \label{tab:loss}%
  \vspace{-0.5em}
  \caption*{{\footnotesize Notes: Table~\ref{tab:loss} reports mean bias and MSE for the estimators across 500 Monte Carlo replications with sample size $n=5,000$. The DGP follows Section~\ref{sec:sim}, expect that (i) the left panel uses $\varepsilon$ drawn from a three-component Gaussian mixture and copula dependence $\rho=0$, while (ii) the right panel set $\varepsilon\stackrel{a.s.}{=}0$ and $\rho=0$. CFQR implements an adaptation of \citet{lee2007endogeneity} with $k=5$ regression splines (to the fourth order). SMLE replicates  \citep{hausman2021errors}. 2SSMLE uses 20 grid knots, B-splines at $r=1$, and $5$ knots per interval for the quadrature. }}%
\end{table}%
\end{landscape}

\pagebreak
\begin{table}[htbp]
  \centering
  \caption{Parametric $\hat{V}$ Bootstrap Results: $\varepsilon\sim 3\mathcal{N}$ and $\rho=0.5$}
    \begin{tabular}{rrc|ccccrr}
    \toprule
    \toprule
    \multicolumn{9}{c}{I. Quantile Coefficient Functions} \\
    \midrule
    \multicolumn{1}{c}{$\hat\beta_0$} & \multicolumn{1}{c}{$\hat{\text{se}}_0$} & coverage & $\hat\beta_1$ & $\hat{\text{se}}_1$ & \multicolumn{1}{c|}{coverage} & $\hat\beta_2$ & \multicolumn{1}{c}{$\hat{\text{se}}_2$} & \multicolumn{1}{c}{coverage} \\
    \midrule
    \multicolumn{1}{c}{1.294} & \multicolumn{1}{c}{0.178} & 0.962 & 1.060 & 0.248 & \multicolumn{1}{c|}{0.970} & 0.314 & \multicolumn{1}{c}{0.059} & \multicolumn{1}{c}{0.948} \\
    \multicolumn{1}{c}{1.524} & \multicolumn{1}{c}{0.179} & 0.962 & 1.225 & 0.167 & \multicolumn{1}{c|}{0.930} & 0.438 & \multicolumn{1}{c}{0.052} & \multicolumn{1}{c}{0.948} \\
    \multicolumn{1}{c}{1.741} & \multicolumn{1}{c}{0.184} & 0.942 & 1.349 & 0.136 & \multicolumn{1}{c|}{0.924} & 0.535 & \multicolumn{1}{c}{0.050} & \multicolumn{1}{c}{0.944} \\
    \multicolumn{1}{c}{1.957} & \multicolumn{1}{c}{0.187} & 0.956 & 1.478 & 0.123 & \multicolumn{1}{c|}{0.936} & 0.622 & \multicolumn{1}{c}{0.047} & \multicolumn{1}{c}{0.940} \\
    \multicolumn{1}{c}{2.185} & \multicolumn{1}{c}{0.185} & 0.946 & 1.620 & 0.115 & \multicolumn{1}{c|}{0.918} & 0.691 & \multicolumn{1}{c}{0.046} & \multicolumn{1}{c}{0.932} \\
    \multicolumn{1}{c}{2.366} & \multicolumn{1}{c}{0.183} & 0.958 & 1.780 & 0.105 & \multicolumn{1}{c|}{0.934} & 0.756 & \multicolumn{1}{c}{0.046} & \multicolumn{1}{c}{0.944} \\
    \multicolumn{1}{c}{2.548} & \multicolumn{1}{c}{0.176} & 0.956 & 1.949 & 0.096 & \multicolumn{1}{c|}{0.944} & 0.818 & \multicolumn{1}{c}{0.046} & \multicolumn{1}{c}{0.928} \\
    \multicolumn{1}{c}{2.688} & \multicolumn{1}{c}{0.163} & 0.956 & 2.149 & 0.082 & \multicolumn{1}{c|}{0.938} & 0.871 & \multicolumn{1}{c}{0.047} & \multicolumn{1}{c}{0.950} \\
    \multicolumn{1}{c}{2.804} & \multicolumn{1}{c}{0.143} & 0.940 & 2.365 & 0.069 & \multicolumn{1}{c|}{0.942} & 0.925 & \multicolumn{1}{c}{0.050} & \multicolumn{1}{c}{0.974} \\
    \midrule
          &       & \multicolumn{1}{c}{} &       &       &       &       &       &  \\
          &       & \multicolumn{5}{c}{II. Distributional Coefficients} &       &  \\
\cmidrule{3-7}          &       &       &       & estimate   & $\hat{\text{se}}$    & coverage &       &  \\
\cmidrule{3-7}          &       & weights & $\lambda_1$ & 0.440 & 0.158 & 0.942 &       &  \\
          &       &       & $\lambda_2$ & 0.177 & 0.168 & 0.984 &       &  \\
          &       &       &       &       &       &       &       &  \\
          &       & means & $\mu_1$   & -2.913 & 1.365 & 0.994 &       &  \\
          &       &       & $\mu_2$   & 0.693 & 2.144 & 0.858 &       &  \\
          &       &       &       &       &       &       &       &  \\
          &       & std   & $\sigma_1$ & 0.850 & 0.432 & 0.962 &       &  \\
          &       &       & $\sigma_2$ & 0.891 & 0.610 & 0.880 &       &  \\
          &       &       & $\sigma_3$ & 1.137 & 0.442 & 0.908 &       &  \\
          &       &       &       &       &       &       &       &  \\
          &       & copula & $\rho$   & 0.468 & 0.040 & 0.982 &       &  \\
\cmidrule{3-7}    \end{tabular}%
  \caption*{{\footnotesize Notes: Table~\ref{tab:pfirstboot} reports the point estimates, bootstrap standard error, and bootstrap coverage across 500 Monte Carlo replications and 200 bootstrap repetitions with sample size $n=5,000$. The DGP follows Section~\ref{sec:sim}, with $\varepsilon$ drawn from a three-component Gaussian mixture and copula dependence $\rho=0.5$. The first stage control function is estimated parametrically. 2SSMLE uses 20 grid knots, B-splines at $r=1$, and $5$ knots per interval for the Gauss-Legendre quadrature.}}%
  \label{tab:pfirstboot}%
\end{table}%

\pagebreak
\begin{table}[htbp]
  \centering
  \caption{Spline $\hat{V}$ Bootstrap Results: $\varepsilon\sim 3\mathcal{N}$ and $\rho=0.5$}
    \begin{tabular}{rrc|ccccrr}
    \toprule
    \toprule
    \multicolumn{9}{c}{I. Quantile Coefficient Functions} \\
    \midrule
    \multicolumn{1}{c}{$\hat\beta_0$} & \multicolumn{1}{c}{$\hat{\text{se}}_0$} & coverage & $\hat\beta_1$ & $\hat{\text{se}}_1$ & \multicolumn{1}{c|}{coverage} & $\hat\beta_2$ & \multicolumn{1}{c}{$\hat{\text{se}}_2$} & \multicolumn{1}{c}{coverage} \\
    \midrule
    \multicolumn{1}{c}{1.133} & \multicolumn{1}{c}{0.213} & 0.958 & 1.160 & 0.145 & \multicolumn{1}{c|}{0.962} & 0.319 & \multicolumn{1}{c}{0.063} & \multicolumn{1}{c}{0.926} \\
    \multicolumn{1}{c}{1.454} & \multicolumn{1}{c}{0.180} & 0.946 & 1.290 & 0.118 & \multicolumn{1}{c|}{0.878} & 0.423 & \multicolumn{1}{c}{0.056} & \multicolumn{1}{c}{0.934} \\
    \multicolumn{1}{c}{1.675} & \multicolumn{1}{c}{0.171} & 0.912 & 1.428 & 0.109 & \multicolumn{1}{c|}{0.832} & 0.521 & \multicolumn{1}{c}{0.056} & \multicolumn{1}{c}{0.930} \\
    \multicolumn{1}{c}{1.883} & \multicolumn{1}{c}{0.166} & 0.916 & 1.570 & 0.104 & \multicolumn{1}{c|}{0.810} & 0.604 & \multicolumn{1}{c}{0.058} & \multicolumn{1}{c}{0.928} \\
    \multicolumn{1}{c}{2.081} & \multicolumn{1}{c}{0.163} & 0.916 & 1.715 & 0.099 & \multicolumn{1}{c|}{0.818} & 0.677 & \multicolumn{1}{c}{0.063} & \multicolumn{1}{c}{0.940} \\
    \multicolumn{1}{c}{2.278} & \multicolumn{1}{c}{0.161} & 0.948 & 1.871 & 0.095 & \multicolumn{1}{c|}{0.830} & 0.743 & \multicolumn{1}{c}{0.069} & \multicolumn{1}{c}{0.938} \\
    \multicolumn{1}{c}{2.447} & \multicolumn{1}{c}{0.159} & 0.952 & 2.043 & 0.088 & \multicolumn{1}{c|}{0.870} & 0.802 & \multicolumn{1}{c}{0.078} & \multicolumn{1}{c}{0.934} \\
    \multicolumn{1}{c}{2.637} & \multicolumn{1}{c}{0.157} & 0.948 & 2.213 & 0.078 & \multicolumn{1}{c|}{0.880} & 0.863 & \multicolumn{1}{c}{0.091} & \multicolumn{1}{c}{0.934} \\
    \multicolumn{1}{c}{2.812} & \multicolumn{1}{c}{0.143} & 0.962 & 2.401 & 0.070 & \multicolumn{1}{c|}{0.882} & 0.923 & \multicolumn{1}{c}{0.117} & \multicolumn{1}{c}{0.936} \\
    \midrule
          &       & \multicolumn{1}{c}{} &       &       &       &       &       &  \\
          &       & \multicolumn{5}{c}{II. Distributional Coefficients} &       &  \\
\cmidrule{3-7}          &       &       &       & estimate   & $\hat{\text{se}}$   & coverage &       &  \\
\cmidrule{3-7}          &       & weights & $\lambda_1$ & 0.464 & 0.114 & 0.944 &       &  \\
          &       &       & $\lambda_2$ & 0.164 & 0.154 & 0.952 &       &  \\
          &       &       &       &       &       &       &       &  \\
          &       & means & $\mu_1$   & -2.970 & 1.022 & 0.994 &       &  \\
          &       &       & $\mu_2$   & -0.200 & 2.044 & 0.780 &       &  \\
          &       &       &       &       &       &       &       &  \\
          &       & std   & $\sigma_1$ & 0.897 & 0.509 & 0.948 &       &  \\
          &       &       & $\sigma_2$ & 1.379 & 0.975 & 0.920 &       &  \\
          &       &       & $\sigma_3$ & 1.159 & 0.394 & 0.884 &       &  \\
          &       &       &       &       &       &       &       &  \\
          &       & copula & $\rho$   & 0.439 & 0.039 & 0.758 &       &  \\
\cmidrule{3-7}    \end{tabular}%
  \caption*{{\footnotesize Notes: Table~\ref{tab:pfirstboot} reports the point estimates, bootstrap standard error, and bootstrap coverage across 500 Monte Carlo replications and 200 bootstrap repetitions with sample size $n=5,000$. The DGP follows Section~\ref{sec:sim}, with $\varepsilon$ drawn from a three-component Gaussian mixture and copula dependence $\rho=0.5$. The first stage control function is estimated with spline regression with $K=6$ knots and power $r=2$. 2SSMLE uses 20 grid knots, B-splines at $r=1$, and $5$ knots per interval for the Gauss-Legendre quadrature.}}%
  \label{tab:npfirstboot}
\end{table}%

\pagebreak
\begin{figure}[H]
	\centering
	\caption{MC Results Comparison: $\varepsilon\sim 3\mathcal{N}$ and $\rho=0.9$} 
	\includegraphics[scale=0.5]{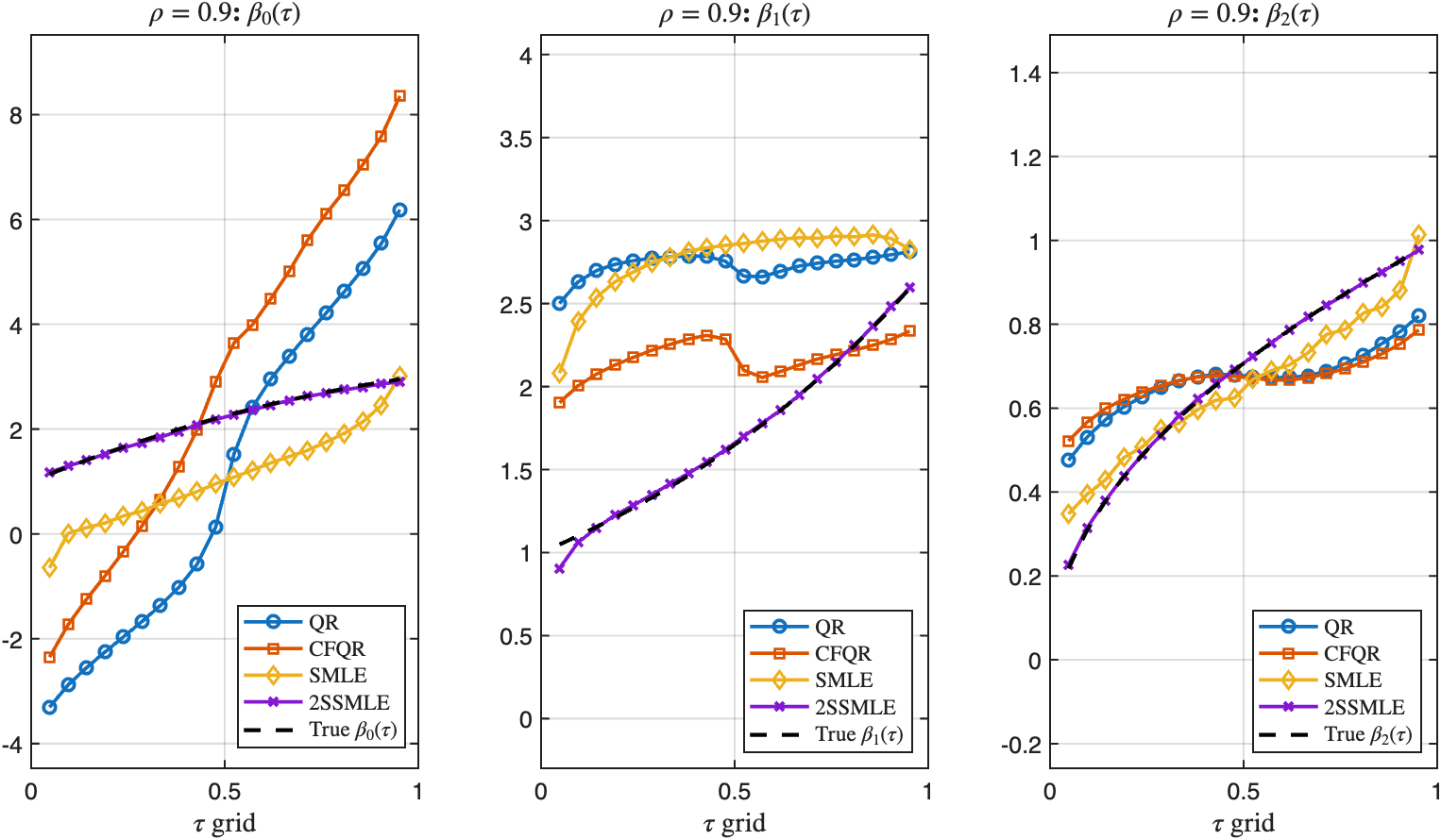}
  \label{fig2}
\end{figure}

\begin{figure}[h]
	\centering 
	\caption{MC Results Comparison: $\varepsilon\sim \mathcal{N}(0,1)$ and $\rho=0.5$}
	\includegraphics[scale=0.5]{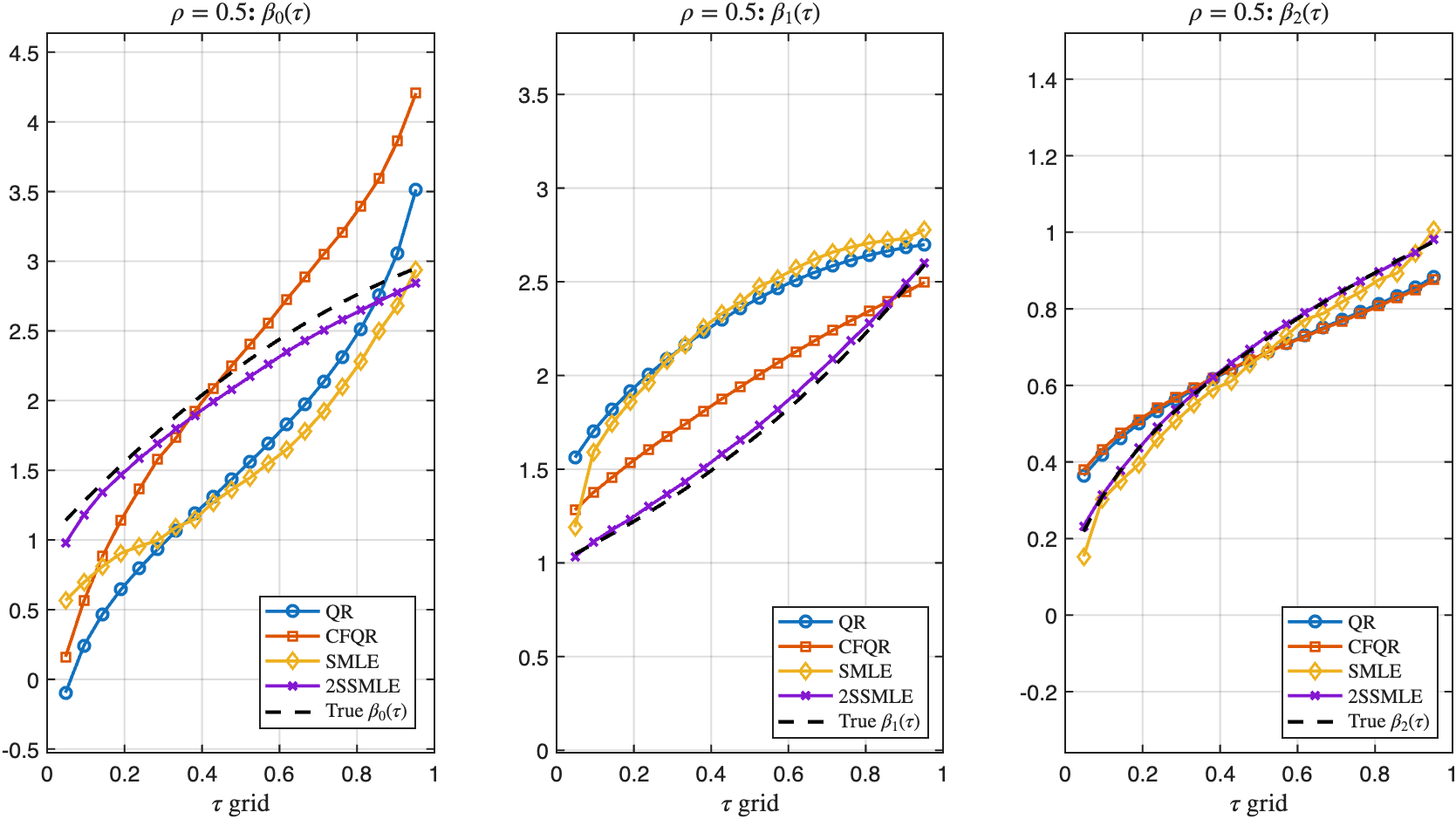}
    \caption*{{\footnotesize Notes: Figure~\ref{fig2} and \ref{fig3} compare the true coefficient functions with four estimators—QR (circle), QR with control function (square), sieve MLE (diamond), and two-step sieve MLE with control function (cross)—based on 500 Monte Carlo replications with sample size $n=5000$. The DGP follows Section~\ref{sec:sim}. Measurement error in \ref{fig2} is drawn from drawn from a three-component Gaussian mixture and copula dependence $\rho=0.9$, while \ref{fig3} is from $\mathcal{N}(0,1)$ and copula dependence $\rho=0.5$. CFQR implements an adaptation of \citet{lee2007endogeneity} with $k=5$ regression splines (to the fourth order). SMLE replicates  \citep{hausman2021errors}. 2SSMLE uses 20 grid knots, B-splines at $r=1$, and $5$ knots per interval for the Gauss-Legendre quadrature.}}
    \label{fig3}
\end{figure}

\newpage
\subsection{Appendix B: Proofs of Results in Section \ref{sec:modelid}}

\begin{proof}[\textbf{Proof of Proposition~\ref{prop1}}]
\label{pf_prop1}
	For any $y\in\mathcal{Y}$, $x\in\mathcal{X}$, and $v\in(0,1)$,
	\begin{align*}
  F_{Y| X,V}(y| x,v)&=P(Y\le y| X=x,V=v)\\
  &=P(x'\beta_0(U)+\varepsilon\le y| X=x,V=v)\\
  &=\int_0^1 P(\varepsilon\le y-x'\beta_0(u)| X=x,V=v,U=v)F_{U|X,V}(du| x,v)\\
  &=\int_0^1 P(\varepsilon\le y-x'\beta_0(u))dF_{U|V}(du| v)\\
  &=\int_0^1 F_\varepsilon(y-x'\beta_0(u))dF_{U|V}(du| v),
\end{align*} where the third equality follows by iterated expectations. The fourth relies on (i) independence of $\varepsilon$ from $(X,V,U)$ and (ii) $U\ind X| V$. Differentiating in $y$ yields \eqref{e5}.
\end{proof}

\begin{proof}[\textbf{Proof of Theorem~\ref{thm1}}]
\label{pf_thm1}
	Since $\varepsilon\ind (X,Z)$ and $V=F_{X_1|Z}(X_1,Z)$, it follows that $\varepsilon$ is independent of $V$. Given $U\ind \varepsilon$ and $\varepsilon\ind (X,V)$, $\varepsilon$ is also independent of $U$ conditional on $(X,V)$. The characteristic function of $Y$ conditional on any $X=x$ and $V=v$ is
	\begin{align*}
		\phi_{Y|X,V}(s;x,v)&=\E[e^{isY}|X=x,V=v]\\
		&=\E[e^{is(x'\beta_0(U)+\varepsilon)}|X=x,V=v]\\
		&=\E[e^{is\varepsilon}|X=x,V=v]\cdot \E[e^{isx'\beta_0(U)}|X=x,V=v]\\
		&=\E[e^{is\varepsilon}]\cdot \E[e^{isx'\beta_0(U)}|V=v]\\
		&=\phi_{\varepsilon}(s)\cdot\phi_{g_x|V}(s;x,v),
	\end{align*} where $g_x(u)=x'\beta_0(u)$ for each $x\in\mathcal{X}$. The first equality holds by the definition of the conditional characteristic function, the second by substituting $Y=X'\beta_0(U)+\varepsilon$, the third by the conditional independence between $\varepsilon$ and $U$ given $(X,V)$, and the fourth by the independence of $\varepsilon$ from $(X,V)$ and $X$ from $U$ given $V$.\\
	
	Suppose there exists another parameter set $(\beta(\cdot),f_\varepsilon,f_{U|V})\in S$ that is observationally equivalent to the true parameters $(\beta_0(\cdot),f_{0,\varepsilon},f_{0,U|V})$. Then, the conditional characteristic function can also be expressed as
	\begin{align}\label{p1}
		\phi_{Y|X,V}(s;x,v)=\phi_{0,\varepsilon}(s)\cdot\phi_{0,g_x|V}(s;x,v)&=\phi_\varepsilon(s)\cdot\phi_{g_x|V}(s;x,v),
	\end{align} which can then be expended as
    \begin{align*}
		\int_{\R} \exp(ise)f_{0,\varepsilon}(e)de\cdot \int_0^1 \exp(is x'\beta_0(u))f_{0,U|V}(u,v)du&=\int_{\R} \exp(ise)f_\varepsilon(e)de\cdot \int_0^1 \exp(is x'\beta(u))f_{U|V}(u,v)du\nonumber\\
    m(s)\cdot\int_0^1 \exp(is x'\beta_0(u))f_{0,U|V}(u,v)du&=\int_0^1 \exp(is x'\beta(u))f_{U|V}(u,v)du,\nonumber
    \end{align*} where $m(s)=\phi_{0,\varepsilon}(s)/\phi_{\varepsilon}(s)$. We take the $x_p$ in Assumption \ref{npx}(b) to be $x_1$ for notation convenience, and wlog, assume the support of $x_1|x_{-1}$ contains an open neighborhood of 0. Applying a Taylor expansion around $x_1=0$ to both sides of the last equality yields 
	\begin{align*}
		m(s)\int_0^1 \exp(isx_{-1}'\beta_{0,-1}(u))\cdot \sum_{k=0}^\infty \frac{(is)^k}{k!}\beta_{0,1}^k(u)x_1^k f_{0,U|V}(u,v)du\\
		=\int_0^1 \exp(isx_{-1}'\beta_{-1}(u))\cdot \sum_{k=0}^\infty \frac{(is)^k}{k!}\beta_1^k(u)x_1^k f_{U|V}(u,v)du.
	\end{align*} Since $x_1$ is continuous, the coefficients of the corresponding polynomials in $x_1$ must be identical. That is, for all $k\ge1$,
	\begin{align}\label{p2}
		m(s)\frac{(is)^k}{k!}\int_0^1 \exp(isx_{-1}'\beta_{0,-1}(u))\beta_{0,1}^k(u)f_{0,U|V}(u,v)du\nonumber\\
		=\frac{(is)^k}{k!}\int_0^1 \exp(isx_{-1}'\beta_{-1}(u))\beta_{1}^k(u)f_{U|V}(u,v)du.
	\end{align} For $s\neq0$, we can cancel out $(is)^k/k!$ on both sides. Since the expression above is continuous in $s$ and $m(0)=1$, letting $s\to0$ gives
	\begin{align}\label{p3}
		\int_0^1 \beta_{0,1}^k(u)f_{0,U|V}(u,v)du=\int_0^1 \beta_1^k(u)f_{U|V}(u,v)du.
	\end{align} By the common support Assumption~\ref{npc} (b), for all $X=x\in\mathcal{X}$, the support of $V|X=x$ remains $[0,1]$. Integrating both sides with respect to $V$ over $[0,1]$, we obtain
	\begin{align*}
		\int_0^1 \beta_{0,1}^k(u)\int_0^1 f_{0,U|V}(u,v)f_V(v)dvdu&= \int_0^1 \beta_{1}^k(u)\int_0^1 f_{U|V}(u,v)f_V(v)dvdu\\
		\int_0^1 \beta_{0,1}^k(u) f_U(u)du&=\int_0^1 \beta_{1}^k(u)f_U(u)du\\
		\int_0^1 \beta_{0,1}^k(u)du &=\int_0^1 \beta_{1}^k(u)du,
	\end{align*} where the third equality follows from $\int_0^1 f_{0,U|V}(u,v)f_V(v)dv=\int_0^1 f_{0,U,V}(u,v)dv=f_{U}(u)=1$. Since this holds for all $k\ge1$, the characteristic functions of $\beta_{0,1}(U)$ and $\beta_{1}(U)$ are
	\begin{align*}
		\phi_{\beta_{0,1}(U)}(s)&=\int_0^1 \exp(is\beta_{0,1}(u))du\\
		&=\sum_{k=0}^\infty \frac{(is)^k}{k!}\int_0^1 \beta_{0,1}^k(u)du\\
		&=\sum_{k=0}^\infty \frac{(is)^k}{k!}\int_0^1 \beta_1^k(u)du\\
		&=\int_0^1 \exp(is\beta_1(u))du=\phi_{\beta_1(U)}(s).
	\end{align*} This implies that $\beta_{0,1}(U)$ and $\beta_{1}(U)$ have the same distribution almost everywhere (a.e). Hence, there exists a measure-preserving transform $\pi:[0,1]\to[0,1]$, such that $\beta_1(\pi(u))=\beta_{0,1}(u)$ a.e. For any integrable $h(\cdot)$ defined on $[0,1]$, $\int_0^1 h(\pi(u))du=\int_0^1 h(u)du$. Substituting $\pi$ into (\ref{p2}) and integrating with respect to $V$,	\begin{align*}
		&m(s)\frac{(is)^k}{k!}\int_0^1 \exp(isx_{-1}'\beta_{0,-1}(u))\beta_{0,1}^k(u)du\\
		&=\frac{(is)^k}{k!}\int_0^1 \exp(isx_{-1}'\beta_{-1}(u))\beta_{1}^k(u)du\\
		&=\frac{(is)^k}{k!}\int_0^1 \exp(isx_{-1}'\beta_{-1}(\pi(u))\beta_{1}^k(\pi(u))du.
	\end{align*} Since $\varepsilon$ has zero mean under both $f_{0,\varepsilon}$ and $f_{\varepsilon}$, we have $m'(0)=0$. Taking the first-order derivative of above with respect to $s$ at $s=0$, we obtain
	\begin{align*}
		\int_0^1 x_{-1}'\beta_{0,-1}(u)\beta_{0,1}^k(u)du=\int_0^1 x_{-1}'\beta_{-1}(\pi(u))\beta_1^k(u)du.
	\end{align*} By Assumption \ref{npx} (c), since $\beta_{0,1}$ is continuous and strictly monotonic, $\{\beta_{0,1}^k\}_{k\ge0}$ is a functional basis of $L^2[0,1]$. Thus $x_{-1}'\beta_{-1}(\pi(u))=x_{-1}'\beta_{0,-1}(u)$ a.e, and $\E[x_{-1}x_{-1}']\beta_{-1}(\pi(u))=\E[x_{-1}x_{-1}']\beta_{0,-1}(u)$ a.e. Since $\E[xx']$ is nonsingular, $\E[x_{-1}x_{-1}']$ is also nonsingular. Multiplying both sides by $\E[x_{-1}x_{-1}']^{-1}$, we have $\beta_{-1}(\pi(u))=\beta_{0,-1}(u)$ a.e. Combining with $\beta_1(\pi(u))=\beta_{0,1}(u)$, we have $\beta(\pi(u))=\beta_0(u)$ a.e, and $x'\beta(\pi(u))=x'\beta_0(u)$ a.e for any $x\in\mathcal{X}$. Note that $x'\beta(\pi(U))$ and $x'\beta(u)$ follow the same distribution, and that $x'\beta(u)$ is strictly monotonic in $u$, thus it must be $x'\beta(u)=x'\beta_0(u)$ a.e, and $\E[xx']\beta(u)=\E[xx']\beta_0(u)$ a.e. Multiplying both sides by $\E[xx']^{-1}$, we conclude that $\beta(u)=\beta_0(u)$ a.e. Suppose $\beta$ deviates from $\beta_0$ at some $\tilde{u}\in[0,1]$, then $\beta$ violates the continuity assumption \ref{npb} (b), if $\beta_0\in M$. Hence, $\beta_0(u)=\beta(u)$ for all $u\in[0,1]$.  \\
	
	Since $\beta_{0,1}=\beta_{1}$, and that $\{\beta_{0,1}^k\}_{k\ge0}$ is a functional basis of $L^2[0,1]$, by equation (\ref{p3}), $f_{0,U|V}=f_{U|V}$ a.e for any given $v\in[0,1]$. With a uniform margin, the joint copula density $f_{0,U,V}=f_{U,V}$ a.e. Based on what we have shown, we get $\phi_{0,g_x|V}(s;x,v)=\phi_{g_x|V}(s;x,v)$ for almost every $x\in\mathcal{X}$ and $v\in[0,1]$, and from (\ref{p1}), $\phi_{0,\varepsilon}(s)=\phi_{\varepsilon}(s)$ a.e, so $f_{0,\varepsilon}=f_\varepsilon$ a.e. Therefore, the true parameter set $(\beta_0(\cdot),f_{0,\varepsilon}.f_{0,U|V})$ is uniquely identified.	
\end{proof}


\subsection{Appendix C: Proofs of Results in Section \ref{subsec:asymp_p}}

\begin{proof}[\textbf{Proof of Theorem~\ref{thm2}}]
\label{pf_thm2}
By Theorem~\ref{thm1}, $\theta_0 = \arg\max_{\theta\in\Theta_0} Q(\theta,\pi_0)$, while the estimator $\hat{\theta}_n^p \in \arg\max_{\theta\in\Theta_n} \hat{Q}_n(\theta,\hat\pi)$. Since every function in $M$ is uniformly bounded and satisfies a H\"{o}lder condition of order $\bar{p}=\max_{1\le k\le d_x} p_k$ with a fixed constant $\bar{c}=\max_{1\le k\le d_x} c_k$, the Arzel\`{a}-Ascoli theorem implies that $M$ is precompact. Since $M$ is closed, it is also compact under $d(\cdot,\cdot)$. Thus, $\Theta_0=M\times\Sigma\times\Gamma$ is compact. Completeness follows from the completeness of $L^p$ spaces and the fact that limits of monotone functions remain monotone.

Let $W=(Y,X',Z')'$. Denote the conditional log-likelihood function at given $(\pi,\theta)$ as $l(w;\theta,\pi):=\log f(y|x,v(x,z;\pi);\theta)$. The true control variable $V_i=V(X_i,Z_i;\pi_0)$. By MVT and CS inequality, there exist $s > 0$ and a random variable $L(W)$ with $\E[|L(W)|^2] < \infty$ (Assumptions~\ref{npx}, \ref{npb}, \ref{parpro}) such that
	\begin{align*}
		\sup_{\theta,\theta'\in\Theta_n:d(\theta,\theta')\le\delta} |l(w;\theta,\pi_0)-l(w;\theta',\pi_0)|\le \delta^s \cdot L(w).
	\end{align*} Let $N(\delta, \{l(w;\theta,\pi_0) : \theta\in\Theta_n\}, L_{1,n})$ denote the $L_1(P_n)$-covering number. It follows that 
	\begin{align*}
		\log N(\delta,\{l(w;\theta,\pi_0):\theta\in \Theta_n\},L_{1,n})\le \log N(\delta^{1/s}, \Theta_n, d(\cdot,\cdot))\le K\cdot J_n\cdot \log(\delta^{-1/s}),
	\end{align*} for some constant $K$ by \citet{chen1998sieve}. Since $J_n=o(n)$, we have $\log N(\delta,\{l(w;\theta,\pi_0):\theta\in \Theta_n\},L_{1,n})=o(n)$ for any fixed $\delta>0$. Applying Theorem 24 of \cite{pollard1989asymptotics}, we thus have $\sup_{\theta\in\Theta_n}|(\E_n-\E)l(w;\theta,\pi_0)|=o_p(1)$. For any fixed $v, \hat{v} \in (0,1)$, Assumption~\ref{parpro} ensures
	\begin{align}
		\sup_{\theta\in \Theta_n}|\hat{Q}_n(\theta,\pi_0)-Q_n(\theta,\hat\pi)|&\le\frac1n\sum_{i=1}^n \sup_{\theta\in \Theta_n}\left|l(W_i;\theta,\pi_0)-l(W_i;\theta,\hat\pi)\right|\nonumber\\
		&\le \Big(\frac{1}{n}\sum_{i=1}^n L'(W_i)^2\Big)^{1/2}\Big(\frac{1}{n}\sum_{i=1}^n (\hat{V}_i-V_i)^2\Big)^{1/2}=o_p(1),\label{p5}
	\end{align} for some $L'(W)=O_p(1)$. By Assumptions~\ref{npe}-\ref{sieve} and DCT, $Q(\theta,\pi_0)$ and $Q_n(\theta,\pi_0)$ are both continuous in $\theta$. For any $\epsilon>0$, w.p.a.1, we have (a) $\hat{Q}_n(\hat\theta_n^p,\hat\pi)>\hat{Q}_n(\theta_J^*,\hat\pi)-\epsilon/6$ by the definition of $\hat\theta_n^p$; (b) $\sup_{\theta\in\Theta_n}|\hat{Q}_n(\theta,\hat\pi)-Q_n(\theta,\pi_0)|<\epsilon/6$ by (\ref{p5}); (c) $\sup_{\theta\in\Theta_n}|Q_n(\theta,\pi_0)-Q(\theta,\pi_0)|<\epsilon/6$ by uniform convergence; and (d) $|Q(\theta_J^*,\pi_0)-Q(\theta_0,\pi_0)|<\epsilon/6$ by the continuity of $Q$. Thus, w.p.a.1, we have
	\begin{align*}
			\hat{Q}_n(\hat\theta_n^p,\hat\pi)&\stackrel{(a)}{>}\hat{Q}_n(\theta_J^*,\hat\pi)-\epsilon/6
			\stackrel{(b)}{>}Q_n(\theta_J^*,\pi_0)-\epsilon/3
			\stackrel{(c)}{>}Q(\theta_J^*,\pi_0)-\epsilon/2
			\stackrel{(d)}{>}Q(\theta_0,\pi_0)-2\epsilon/3\\
			\hat{Q}_n(\hat\theta_n^p,\hat\pi)&\stackrel{(b)}{<}Q_n(\hat\theta_n^p,\pi_0)+\epsilon/6
			\stackrel{(c)}{<}Q(\hat\theta_n^p,\pi_0)+\epsilon/3.
		\end{align*} Thus, for any $\epsilon>0$, $Q(\theta_0,\pi_0)<Q(\hat\theta_n,\pi_0)+\epsilon$ w.p.a.1. By our identification Theorem~\ref{thm1} and a similar argument as \citep{newey1994large}, $d(\hat\theta_n^p,\theta_0)\xrightarrow{p}0$.
\end{proof}

\begin{proof}[\textbf{Proof of Lemma~\ref{lem1}}]
\label{pf_lem1}
In the following proof, we suppress $v(x,z;\pi_0)$ as $v$, and $\hat{v}(x,z;\hat\pi)$ as $\hat{v}$ in notation. Consider the function class $\mathcal{G}=\{\log f(y|x,v;\beta,\sigma,\lambda):\beta\in M,\sigma\in\Sigma,\gamma\in\Gamma\}$. The bracketing entropy integral of $\mathcal{G}$ is 
	\begin{align*}
		J_{[]}(\delta,\mathcal{G},L_2(P))&=\int_0^\delta \sqrt{\log N_{[]}(\epsilon,\mathcal{G},L_2(P))}d\epsilon\le 
	\int_0^\delta \sqrt{\log N_{[]}(\epsilon^{1/s},\Theta_0,d(\cdot))}d\epsilon\\ 
	&\le \int_0^\delta \sqrt{ K(\epsilon^{-1/m s})-2\log(\epsilon^{1/s})}d\epsilon<\infty,
	\end{align*} where the second inequality is given by the smoothness of $\beta\in M$ at order $m$, and the final inequality holds for all $ms>1$. Hence, $\mathcal{G}$ is $P$-Donsker. 
	
	Part (a) of Lemma~\ref{lem1} follows closely from Lemma 4 in \cite{hausman2021errors}. Their proof can be adapted to accommodate the control variable $v$ with the following modification:
	\begin{align*}
		Var\left(\frac{\phi_{x\hat\beta_n,\hat\gamma}(s|x,v)}{\phi_{x\beta_0,\gamma_0}(s|x,v)}\right)&=Var \left(\frac{\phi_{x\hat\beta_n,\hat\gamma}(s|x,v)-\phi_{x\hat\beta_n,\gamma_0}(s|x,v)}{\phi_{x\beta_0,\gamma_0}(s|x,v)}+\frac{\phi_{x\hat\beta_n,\gamma_0}(s|x,v)-\phi_{x\beta_0,\gamma_0}(s|x,v)}{\phi_{x\beta_0,\gamma_0}(s|x,v)}\right)\\
		&\ge K\cdot Var\left(\frac{\phi_{x\hat\beta_n,\gamma_0}(s|x,v)-\phi_{x\beta_0,\gamma_0}(s|x,v)}{\phi_{x\beta_0,\gamma_0}(s|x,v)}\right)\\
		&\ge cK\cdot \E\left[\left|\frac{\phi_{x\hat\beta_n,\gamma_0}(s|x,v)-\phi_{x\beta_0,\gamma_0}(s|x,v)}{\phi_{x\beta_0,\gamma_0}(s|x,v)}\right|^2\right],
	\end{align*} where the last step invokes Assumption~\ref{cf}. Along with $J_n^{2r+2}/n\to\infty$, we analogously have $\|\hat\sigma-\sigma_0\|_2^2=O_p(\delta\sqrt{-\log\delta}/\sqrt{n})$, where $\delta=\max\{\|\hat\beta_n-\beta_J^*\|_\infty,\|\hat\gamma-\gamma_0\|_2,\|\hat\sigma-\sigma_0\|_2\}$. If $\delta=\|\hat\sigma-\sigma_0\|_2$, then $\|\hat\sigma-\sigma_0\|_2^2=O_p(\log n/n)$. Otherwise, the conclusion holds.
	
		The argument for part (b) closely parallels Lemma 9 in \cite{hausman2021errors}. Let $G_n := \sqrt{n}(\mathbb{P}_n - \mathbb{P})$. Fix $\gamma = \hat{\gamma}$ and $\sigma = \sigma_0$. By the maximal inequality,
	\begin{align*}
		\E\left[\max_{\|\beta-\beta_J^*\|_2<\delta} \left|
		G_n \log f(y|x,v;\beta,\sigma_0,\hat\gamma)-
		G_n \log f(y|x,v;\beta_J^*,\sigma_0,\hat\gamma)\right| \right]\le \delta\sqrt{-\log \delta}.
	\end{align*} Thus, we have $|G_n \log f(y|x,v;\hat\beta_n,\sigma_0,\hat\gamma)-
		G_n \log f(y|x,v;\beta_J^*,\sigma_0,\hat\gamma)|=O_p( \hat\delta\sqrt{-\log \hat\delta})$, where $\hat\delta=\|\hat\beta_n-\beta_J^*\|_2$. From (a), $\|\hat\sigma-\sigma_0\|_2^2=O_p( \hat\delta\sqrt{-\log \hat\delta/n})$. Hence,
		\begin{align}
			&\E[\log f(y|x,v;\beta_J^*,\sigma_0,\hat\gamma)]-\E[\log f(y|x,v;\hat\beta_n,\sigma_0,\hat\gamma)]\nonumber\\
			= &\frac{1}{\sqrt{n}}G_n[\log f(y|x,v;\hat\beta_n,\sigma_0,\hat\gamma)]-\frac{1}{\sqrt{n}} G_n [\log f(y|x,v;\beta_J^*,\sigma_0,\hat\gamma)]\label{c1}\\
			+ & \E_n [\log f(y|x,v;\beta_J^*,\sigma_0,\hat\gamma)]-\E_n[\log f(y|x,\hat{v};\hat\beta_n,\hat\sigma,\hat\gamma)]\label{c2}\\
			+ &\E_n[\log f(y|x,\hat{v};\hat\beta_n,\hat\sigma,\hat\gamma)]-\E_n[\log f(y|x,\hat{v};\hat\beta_n,\sigma_0,\hat\gamma)]\label{c3}\\
			+ &\E_n[\log f(y|x,\hat{v};\hat\beta_n,\sigma_0,\hat\gamma)]-
			\E_n[\log f(y|x,v;\hat\beta_n,\sigma_0,\hat\gamma)].\label{c4}
		\end{align} Note that term (\ref{c1}) is $O_p( \hat\delta\sqrt{-\log \hat\delta/n})$. Term (\ref{c2}) $\le0$ by the definition of $\hat\theta_n$. Term (\ref{c3}) is $O_p(\|\hat\sigma-\sigma_0\|^2_2)=O_p( \hat\delta\sqrt{-\log \hat\delta/n})$, with the first-order term vanished. (\ref{c4}) is $O_p(\|\hat{v}-v\|)=O_p(n^{-1/2})$ by the first-stage rate. By Theorem~\ref{thm2} and the approximation property, we have $\|\beta_J^*-\beta_0\|_2=O(\frac{1}{J_n^{r+1}})\to0$, $\hat\delta\le \|\hat\beta_n-\beta_0\|_2+\|\beta_J^*-\beta_0\|_2\to0$. Hence, term (\ref{c4}) dominates all others, yielding 		\begin{align*}
			\E[\log f(y|x,v;\beta_J^*,\sigma_0,\hat\gamma)]-\E[\log f(y|x,v;\hat\beta_n,\sigma_0,\hat\gamma)]=O_p(\frac{1}{\sqrt{n}}).
		\end{align*} 
		Denote the density difference $h(y|x,v):= f(y|x,v,\hat\beta_n,\sigma_0,\hat\gamma)-f(y|x,v,\beta_J^*,\sigma_0,\hat\gamma)$. Let $\|h(y|x,v)\|_1:=\int_\R |h(y|x,v)|dy$ denote the $L^1$-norm conditional on $(x,v)$. Applying Pinsker's inequality gives
				\begin{align*}
			\E_{x,v}\left[\|h(y|x,v)\|_1^2 \right]&\le \E_{x,v}\left[K D_{KL}(f(y|x,v,\hat\beta_n,\sigma_0,\hat\gamma)||f(y|x,v,\beta_J^*,\sigma_0,\hat\gamma)) \right]\\
			&\le \E_{x,v}\left[K \int_{\R}f(y|x,v,\hat\beta_n,\sigma_0,\hat\gamma)-f(y|x,v,\beta_J^*,\sigma_0,\hat\gamma) dy\right]\\
			&=K\E[\log f(y|x,v;\beta_J^*,\sigma_0,\hat\gamma)]-K\E[\log f(y|x,v;\hat\beta_n,\sigma_0,\hat\gamma)]=O_p(\frac{1}{\sqrt{n}}).
		\end{align*} As shown in Theorem~\ref{thm1} proof, the difference in conditional characteristic functions satisfies
		\begin{align*}
			|\phi_{x\hat\beta_n,\hat\gamma}(s|x,v)\phi_\varepsilon(s;\sigma_0)-\phi_{x\beta_J^*,\hat\gamma}(s|x,v)\phi_\varepsilon(s;\sigma_0)|=|\int_\R h(y|x,v)\exp(isy)dy|\le \|h(y|x,v)\|_1.
		\end{align*} Therefore, 
		\[
	\bigl| \phi_{x\hat{\beta}_n,\hat{\gamma}}(s|x,v) - \phi_{x\beta_J^*,\hat{\gamma}}(s|x,v) \bigr| \le \frac{\|h(y|x,v)\|_1}{|\phi_\varepsilon(s;\sigma_0)|}.
	\] Under Assumption~\ref{parpro} (c), one can show that $\max\bigl\{ |\phi_{x\hat{\beta}_n,\hat{\gamma}}(s|x,v)|,\ |\phi_{x\beta_J^*,\hat{\gamma}}(s|x,v)| \bigr\} \le K J_n / |s|$, for any fixed $(x,v)\in\mathcal{X}\times \mathcal{V}$. Since $\phi_\varepsilon$ is ordinary smooth, $\phi_\varepsilon(s;\sigma_0)^{-1}\asymp s^\lambda$. Applying the inverse Fourier representation, for some $q>0$,
		\begin{align*}
			&\E_{x,v}\left[\left|F_{x\hat\beta_n,v}(t)-F_{x\beta_J^*,v}(t)\right|\right]\\
			\le & \E_{x,v}\left[\left|\int_{-q}^q\left|\frac{\exp(its)}{2\pi is}\right|\frac{\|h(y|x,v)\|_1}{|\phi_\varepsilon(s;\sigma_0)|} \right|ds\right]
			+\E_{x,v}\left[2\int_q^\infty \frac{1}{2\pi s}|\phi_{x\hat\beta_n,\hat\gamma}(s|x,v)-\phi_{x\beta_J^*,\hat\gamma}(s|x,v)|ds\right]\\
			\quad\le & \frac{1}{2\pi} \int_{-q}^q \frac{1}{|s| |\phi_\varepsilon(s;\sigma_0)|}\, ds \cdot \E_{x,v}[\|h(y|x,v)\|_1]
			+\frac{2}{\pi}\int_q^\infty \frac{J_n}{s^2}\, ds 
		\precsim_{p}  q^\lambda\cdot n^{-1/4}+J_n/q.
		\end{align*} Balancing the two terms by setting $q^\lambda n^{-1/4} \asymp J_n / q$ gives the optimal cutoff $q^* \asymp J_n^{\frac{1}{\lambda+1}} n^{\frac{1}{4(\lambda+1)}}$. Substituting back produces the $L^2$-type bound $\E_{x,v}\bigl[ |F_{x\hat{\beta}_n,v}(t) - F_{x\beta_J^*,v}(t)| \bigr] = O_p\bigl( J_n^{\frac{\lambda}{\lambda+1}} n^{-\frac{1}{4(\lambda+1)}} \bigr)$. Since $\|\hat{\beta}_n - \beta_J^*\|_2$ is stochastically bounded by the above quantity, we deduce that 
		\[
		\hat\delta=\|\hat\beta_n-\beta_J^*\|_2=O_p(\E_{x,v}\left[\left|F_{x\hat\beta_n,v}(t)-F_{x\beta_J^*,v}(t)\right|\right])=O_p(J_n^{\frac{\lambda}{\lambda+1}}\cdot n^{\frac{-1}{4(\lambda+1)}}),
		\] which echos the desired $L^2$ rate. This completes the proof of part (b).
		
		We now prove part (c) of Lemma~\ref{lem1}. Fix $\beta = \hat{\beta}_n$ and $\sigma = \sigma_0$. By the same maximal inequality and first-stage uniform rate arguments applied in the proof of part (b), we obtain
		\begin{align}
	\E\bigl[ \log f(y|x,v;\hat{\beta}_n,\sigma_0,\hat{\gamma}) \bigr] 
- \E\bigl[ \log f(y|x,v;\hat{\beta}_n,\sigma_0,\gamma_0) \bigr] 
= O_p(n^{-1/2}).
	\label{eq:c-op}
	\end{align} Define $\tilde{h}(y|x,v):=f(y|x,v;\hat\beta_n,\sigma_0,\hat\gamma)-f(y|x,v;\hat\beta_n,\sigma_0,\gamma_0)$. Applying Pinsker's inequality in the same manner as in part (b) yields
\begin{align*}
\E_{x,v}\bigl[ \|\tilde{h}(y|x,v)\|_1^2 \bigr] 
&\le K \E_{x,v}\bigl[ D_{\mathrm{KL}}\bigl( f(\,\cdot\,|x,v;\hat{\beta}_n,\sigma_0,\hat{\gamma}) \big\| f(\,\cdot\,|x,v;\hat{\beta}_n,\sigma_0,\gamma_0) \bigr) \bigr] \\
&\le K \E\bigl[ \log f(y|x,v;\hat{\beta}_n,\sigma_0,\gamma_0) - \log f(y|x,v;\hat{\beta}_n,\sigma_0,\hat{\gamma}) \bigr] = O_p(n^{-1/2}).
\end{align*}
	Since $\hat\beta_n\in\mathcal{N}(\beta_0)$ w.p.a.1 by consistency, we apply Assumption~\ref{parpro} (f) to get
		\begin{align*}
			\|\hat\gamma-\gamma_0\|_2^2&\le K \cdot \E_{x,v}\left[
			\int_{-l}^l |\phi_{x\hat\beta_n,\hat\gamma}(s|x,v)-\phi_{x\hat\beta_n,\gamma_0}(s|x,v)|^2ds\right]\\
			&\le K\cdot \E_{x,v}\left[\int_{-l}^l \frac{\|\tilde{h}(y|x,v)\|_1^2}{\phi_\varepsilon^2(\varepsilon;\sigma_0)} ds\right]\\
			&\le K\cdot \E_{x,v}[\|\tilde{h}(y|x,v)\|_1^2],
		\end{align*} for some $0<l<\infty$. Therefore, we deduce that $\|\hat\gamma-\gamma_0\|_2^2=O_p(n^{-1/2})$.
\end{proof}

\pagebreak
\begin{proof}[\textbf{Proof of Theorem~\ref{thm3}}]
\label{pf_thm3}
We suppress the notation $\hat\theta_n^p$ as $\hat\theta_n$. Denote the score function as $\psi_i(\theta,\pi)\equiv \psi(W_i;\theta,\pi):=\partial_\theta \log f(Y_i|X_i,V(X_i,Z_i;\pi);\theta)$. Since $\hat{\theta}_n = (\hat{\beta}_n(\cdot), \hat{\sigma}, \hat{\gamma})$ maximizes the plug-in empirical criterion $\hat{Q}_n(\theta,\hat\pi)$, the first-order condition implies $\E_n[\psi_i(\hat\theta_n,\hat\pi)]=0$. For notational simplicity, we show the case when $d_x=1$. Extension to $d_x>1$ follows by expanding $\beta(u)$ into a $d_xJ_n$-by-1 vector. Perform a Taylor expansion of the score around $\pi_0$,
	\begin{align*}
		0&=\underbrace{\frac{1}{\sqrt{n}}\cdot\sqrt{n}(\E_n-\E)\psi_i(\hat\theta_n,\pi_0)}_{\mathbb{A}}\\
		&+\underbrace{\E\psi_i(\hat\theta_n,\pi_0)-\E\psi_i(\theta_J^*,\pi_0)+\E\psi_i(\theta_J^*,\pi_0)-\E\psi_i(\theta_0,\pi_0)}_{\mathbb{B}}\\
		&+\underbrace{\E_n[\partial_\pi\psi_i(\hat\theta_n,\pi_0)(\hat\pi-\pi_0)]+o_p(n^{-1/2})}_{\mathbb{C}},
	\end{align*} where the expansion in $\mathbb{B}$ uses $\E\psi_i(\theta_0,\pi_0)=0$ and we denote $G_n=\sqrt{n}(\E_n-\E)$. Since the function class $\mathcal{G}=\{\log f(y|x,v;\theta,\pi_0):\theta\in\Theta_0\}$ is $P$-Donsker and $\hat\theta_n$ is consistent, we have
	\begin{align*}
		\mathbb{A}=\frac{1}{\sqrt{n}} G_n \bigl[ \psi_i(\hat{\theta}_n, \pi_0) \bigr] = \frac{1}{\sqrt{n}} (1 + o_p(1)) G_n \bigl[ \psi_i(\theta_0, \pi_0) \bigr].
	\end{align*} The centered score process converges weakly to a Gaussian limit with covariance $\I(\theta_0,\pi_0)=\E[\psi_i(\theta_0,\pi_0)\psi_i'(\theta_0,\pi_0)]$. By a second-order Taylor expansion around $\theta_J^*$, term $\mathbb{B}$ becomes
	\begin{align}
		\mathbb{B}&=\E\psi_i(\hat\theta_n,\pi_0)-\E\psi_i(\theta_J^*,\pi_0)+\E\psi_i(\theta_J^*,\pi_0)-\E\psi_i(\theta_0,\pi_0) \nonumber \\
		&=\E\bigl[ \partial_\theta \psi_i(\theta_J^*, \pi_0) \bigr] (\hat{b}_n - b_J^*, \hat{\sigma} - \sigma_0, \hat{\gamma} - \gamma_0)' \nonumber \\
		&\quad +O_p(\|\hat{b}_n-b_J^*\|_2^2+ \|\hat\sigma-\sigma_0\|_2^2+\|\hat\gamma-\gamma_0\|_2^2)+O_p(\|\theta_J^*-\theta_0\|_2)\nonumber\\
		&=\mathcal{H}(\theta_J^*, \pi_0) (\hat{b}_n - b_J^*, \hat{\sigma} - \sigma_0, \hat{\gamma} - \gamma_0)'\nonumber\\
		&\quad + O_p\bigl( \|\hat{b}_n - b_J^*\|_2^2 + \|\hat{\gamma} - \gamma_0\|_2^2 \bigr) + O\bigl( J_n^{-(r+1)} \bigr),
	\end{align} where the last line relies on $\|\theta_J^* - \theta_0\|_2 = O(J_n^{-(r+1)})$ from the approximation property of the spline sieves. The smallest eigenvalue of $\mathcal{H}(\theta_J^*, \pi_0)$ is of the same order as that of $\mathcal{H}(\theta_0, \pi_0)$. Using the information identity $\mathcal{H}(\theta_0, \pi_0) = -\I(\theta_0, \pi_0)$ and the assumption that $\min\operatorname{eig}(\I(\theta_0, \pi_0)) \asymp \kappa_{J_n}$, we have $\min\operatorname{eig}\bigl( -\mathcal{H}(\theta_J^*, \pi_0) \bigr) \asymp\kappa_{J_n}$. The condition $\kappa_{J_n}^{-1} J_n^{\frac{\lambda}{\lambda+1}} n^{-\frac{1}{4(\lambda+1)}} \to 0$ together with the $L^2$ rate from Lemma~\ref{lem1}(b) implies that $\|\hat\beta_n-\beta_J^*\|_2^2=\kappa_{J_n}\cdot o_p(\|\hat\beta_n-\beta_J^*\|)$, and $\|\hat{b}_n-b_J^*\|_2^2=\kappa_{J_n}\cdot o_p(\|\hat{b}_n-b_J^*\|)$. A similar relation holds for $\|\hat{\gamma} - \gamma_0\|_2^2$. Thus the quadratic terms are asymptotically negligible compared to the leading linear term scaled by $\kappa_{J_n}$. 
	
	Next, under Assumptions~\ref{parpro} and \ref{firstasmp_p}, 
	\begin{align*}
		\sqrt{n}\mathbb{C}&= \frac{1}{n}\sum_{i=1}^n \partial_v\psi_i(\hat\theta_n,\pi_0)\partial_\pi V(X_i,Z_i;\pi_0)\times \sqrt{n}(\hat\pi-\pi_0)+o_p(1)\\
		&=\frac{1}{n}\sum_{i=1}^n \partial_v\partial_\theta l(W_i;\hat\theta_n,\pi_0)f_\eta(\eta_i)(-Z_i')\times \E[Z_iZ_i']^{-1}\frac{1}{\sqrt{n}}\sum_{i=1}^n Z_i\eta_i+o_p(1)\\
		&=\E[\partial_v\psi_i(\theta_0,\pi_0)\partial_\pi V(X_i,Z_i;\pi_0)]\times\E[Z_iZ_i']^{-1}\frac{1}{\sqrt{n}}\sum_{i=1}^n Z_i\eta_i+o_p(1),
	\end{align*} where the last equality holds by
	\begin{align*}
		&\E_n\partial_\pi \psi_i(\hat\theta_n,\pi_0)-\E \partial_\pi \psi_i(\theta_0,\pi_0)\\
		=& (\E_n-\E)(\partial_\pi\psi_i(\hat\theta_n)-\partial_\pi\psi_i(\theta_0))
		+(\E_n-\E)\partial_\pi \psi_i(\theta_0)+\E[\partial_\pi\psi(\hat\theta_n)-\partial_\pi\psi(\theta_0)]\\
		=& (\E_n-\E)(\partial_v\psi(\hat\theta_n)-\partial_v\psi(\theta_0))\partial_\pi V_i(\pi_0)+ o_p(1)+o_p(1)=o_p(1),
	\end{align*} since $\hat\theta_n$ is consistent to $\theta_0$, and $\|(\E_n-\E)(\partial_v\psi(\hat\theta_n)-\partial_v\psi(\theta_0))\partial_\pi V_i(\pi_0)\|_2=O_p(\sqrt{\frac{J_n}{n}})=o_p(1)$ by $J_n=o(n)$ along with Assumptions~\ref{parpro}-\ref{firstasmp_p}. Define $G_\pi:=\E[\partial_v\psi_i(\theta_0,\pi_0)\partial_\pi V(X_i,Z_i;\pi_0)]$ and $\Sigma_\pi:=\E[\eta_i^2]\E[Z_iZ_i']^{-1}$. Combining $\mathbb{A}$, $\mathbb{B}$ and $\mathbb{C}$, while undersmoothing the sieve approximation bias with $J_n^{2r+2}/n \to \infty$, we arrive at
\begin{gather*}
		- \sqrt{n}\mathcal{H}(\theta_J^*, \pi_0) (\hat{b}_n - b_J^*, \hat{\sigma} - \sigma_0, \hat{\gamma} - \gamma_0)' = G_n \bigl[ \psi(\theta_0, \pi_0) \bigr]+ G_\pi\cdot\E[Z_iZ_i']^{-1}\cdot \frac{1}{\sqrt{n}}\sum_{i=1}^n Z_i\eta_i,\\
		- \sqrt{n\kappa_{J_n}} (\hat{b}_n - b_J^*, \hat{\sigma} - \sigma_0, \hat{\gamma} - \gamma_0)'= \sqrt{\kappa_{J_n}}\mathcal{H}(\theta_J^*, \pi_0)^{-1}\Big[G_n \bigl[ \psi(\theta_0, \pi_0) \bigr]+ G_\pi\cdot\E[Z_iZ_i']^{-1}\cdot \frac{1}{\sqrt{n}}\sum_{i=1}^n Z_i\eta_i \Big].
	\end{gather*} Variance of the right hand side is then
	\begin{align*}
	\Omega_J:=\kappa_{J_n}\mathcal{H}(\theta_J^*,\pi_0)^{-1}\Big[\mathcal{I}(\theta_0,\pi_0)+ G_\pi\Sigma_\pi G'_\pi+2\E\big[\psi(\theta_0,\pi_0)\eta_iZ_i'\E[Z_iZ_i']^{-1}G'_\pi\big]
	\Big]\mathcal{H}(\theta_J^*, \pi_0)^{-1}.
	\end{align*} Consider the first part $\Omega_{\text{sec}}:=\kappa_{J_n}\mathcal{H}(\theta_J^*, \pi_0)^{-1}\mathcal{I}(\theta_0,\pi_0)\mathcal{H}(\theta_J^*, \pi_0)^{-1}$ within $\Omega_J$. Note that
	\begin{align*}
		&\|\kappa_{J_n}\mathcal{H}(\theta_J^*, \pi_0)^{-1}\mathcal{I}(\theta_0,\pi_0)\mathcal{H}(\theta_J^*, \pi_0)^{-1}-\kappa_{J_n}\mathcal{I}(\theta_0,\pi_0)^{-1}\|_2\\
		\le &\: \kappa_{J_n}\|\mathcal{H}(\theta_J^*, \pi_0)^{-1}-\mathcal{H}(\theta_0, \pi_0)^{-1}\|^2_2\cdot\|\mathcal{I}(\theta_0,\pi_0)\|_2\\
		 +&\: 2\kappa_{J_n}\|(\mathcal{H}(\theta_J^*, \pi_0)^{-1}-\mathcal{H}(\theta_0, \pi_0)^{-1})\mathcal{I}(\theta_0,\pi_0)\mathcal{H}(\theta_0, \pi_0)^{-1}\|_2\\
		 +&\: \|\kappa_{J_n}(-\mathcal{H}(\theta_0, \pi_0))^{-1}\mathcal{I}(\theta_0,\pi_0)(-\mathcal{H}(\theta_0, \pi_0))^{-1}-\kappa_{J_n}\mathcal{I}(\theta_0,\pi_0)^{-1}\|_2\to0
	\end{align*} given that $\mathcal{I}(\theta_0,\pi_0)(-\mathcal{H}(\theta_0, \pi_0))^{-1}=\mathbf{I}_{J_n}$, and that
	\begin{align*}
		\|\H(\theta_0,\pi_0)\H(\theta_n^*,\pi_0)^{-1}-\mathbf{I}_{J_n}\|_2&\le \|\H(\theta_0,\pi_0)-\H(\theta_n^*,\pi_0)\|_2\|\H(\theta_n^*,\pi_0)^{-1}\|_2\to0
	\end{align*} by $O(\frac{1}{J_n^{r+1}})=o(\kappa_{J_n})$	. As $\kappa_J\|\mathcal{I}(\theta_0, \pi_0)^{-1}\|_2=1$, $\Omega_{\text{sec}}$ has bounded eigenvalues. Consider the second term $\Omega_{\text{fir}}=\kappa_{J_n}\mathcal{H}(\theta_J^*,\pi_0)^{-1}G_\pi\Sigma_\pi G'_\pi\mathcal{H}(\theta_J^*,\pi_0)^{-1}$. For any nonzero vector $\nu\in\R^{J_n}$,
	\begin{align*}
		\nu'G_\pi\Sigma_\pi G'_\pi\nu &=\sup_{q\in\R^{d_z}:\|q\|_2=1}(\nu'G_\pi\Sigma_\pi^{1/2}q)^2\\
		&= \sup_{q\in\R^{d_z}:\|q\|_2=1} \Big(\E[\nu'\E[\partial_v\psi(\theta_0,\pi_0)|X_i,V_i] f_\eta(\eta_i)Z_i'\Sigma_\pi^{1/2}q]\Big)^2\\
		&\le \sup_{q\in\R^{d_z}:\|q\|_2=1} \Big(\E[\nu'\E[K_s\cdot \psi(\theta_0,\pi_0)|X_i,V_i] f_\eta(\eta_i)Z_i'\Sigma_\pi^{1/2}q]\Big)^2\\
		&\le K^2_s\cdot\E[(\nu'\cdot \psi(\theta_0,\pi_0))^2]\cdot \sup_{q\in\R^{d_z}:\|q\|_2=1}\E[(f_\eta(\eta_i)Z_i'\Sigma_\pi^{1/2}q)^2]\\
		&\le K_s^2\cdot K_c^2\cdot \nu'\E[\psi(\theta_0,\pi_0)\psi'(\theta_0,\pi_0)]\nu =K\cdot \nu'\I(\theta_0,\pi_0)\nu,
	\end{align*} for some constants $K_s,K_c>0$, where the penultimate inequality is by CS, the last one is guaranteed by Assumption~\ref{firstasmp_p}, and the first inequality is from
	\begin{align*}
		\E[\partial_v\psi(\theta_0,\pi_0)|x,v]=-\E[\psi(\theta_0,\pi_0)\frac{\partial_v f(y|x,v;\theta_0)}{f(y|x,v;\theta_0)}|x,v].
	\end{align*} Hence, $0\le G_\pi\Sigma_\pi G'_\pi\le K\cdot\I(\theta_0,\pi_0)$ for some constant $K$. Similarly, we can show that the same for the covariance term. Consequently, we have $0\le\Omega_{\text{sec}}\le \Omega_J\le (1+K)\Omega_{\text{sec}}$, which ensures that $\Omega_J$ also has bounded eigenvalues as $\Omega_{\text{sec}}$. Denote the submatrix of $\Omega_J$ for $\sigma, \gamma$ as $\Omega_{J,\sigma}$, $\Omega_{J,\gamma}$. For any $\nu\in\R^{d_\sigma}$ with $\|\nu\|_2=1$, denote $\xi_n:=\nu'\sqrt{\kappa_{J_n}}\H^{-1}_\sigma(\hat\theta_J^*,\pi_0)G_n\phi_i(\theta_0,\pi_0)$, where $\H^{-1}_\sigma(\hat\theta_J^*,\pi_0)$ is the $d_\sigma$-by-$J_n$ matrix stacking the rows corresponding to $\sigma_0$, and $\phi_i(\theta_0,\pi_0)$ denotes $\psi_i(\theta_0,\pi_0)+G_\pi\E[Z_iZ_i]]^{-1}Z_i\eta_i$. Let $\sigma^2_\xi=Var(\xi_n)=\nu'\Omega_{J,\sigma}\nu$, which is thus within some $[K_1\kappa_{J_n},K_2]$. For any $\epsilon>0$,
	\begin{align*}
		\frac{1}{\sigma_{\xi}^2}\E\big[\xi_n^2\1\{\xi_n>\epsilon \sqrt{n}\sigma_{\xi}\}\big]&\le \frac{1}{\sigma_{\xi}^2}\E\Big[\|\phi_i(\theta_0,\pi_0)\|^2_2\1\{\|\phi_i(\theta_0,\pi_0)\|^2_2>K_2^2\epsilon^2 n \sigma_\xi^2\}\Big]\\
		&\le \frac{1}{K_2^2\epsilon^2 n\sigma_\xi^4}\E\Big[\|\phi_i(\theta_0,\pi_0)\|^4_2\Big]\le \frac{KJ_n^2}{K_2^2\epsilon^2 n\sigma_\xi^4},
	\end{align*} given that $J_n^{\frac{\lambda}{\lambda+1}}n^{\frac{-1}{4(\lambda+1)}}/\kappa_{J_n}\to0$. Therefore, by the Lindeberg-Feller Central Limit Theorem, 
	\begin{align*}
		\Omega_{J,\sigma}^{-1/2}\sqrt{n\kappa_{J_n}}(\hat\sigma-\sigma)=\Omega_{J,\sigma}^{-1/2}\sqrt{\kappa_{J_n}}\H^{-1}_\sigma(\hat\theta_J^*,\pi_0)G_n\phi_i(\theta_0,\pi_0)\xrightarrow{d} \mathcal{N}(0,\mathbf{I}_{d_\sigma}).
	\end{align*} Similarly, one can get for the copula parameters $\hat\gamma$:
	\begin{align*}
		\Omega_{J,\gamma}^{-1/2}\sqrt{n\kappa_{J_n}}(\hat\gamma-\gamma)=\Omega_{J,\gamma}^{-1/2}\sqrt{\kappa_{J_n}}\H^{-1}_\gamma(\hat\theta_J^*,\pi_0)G_n\phi_i(\theta_0,\pi_0)\xrightarrow{d} \mathcal{N}(0,\mathbf{I}_{d_\gamma}).
	\end{align*} For the quantile coefficient at any fixed $\tau\in(0,1)$,
	\begin{align*}
		\sqrt{n\kappa_{J_n}}(\hat\beta_n(\tau)-\beta_J^*(\tau))&=\sqrt{n\kappa_{J_n}}(\hat{b}_n-b_J^*)'S(\tau)\\
		&=\sqrt{\kappa_{J_n}}S'(u)\H_b(\theta_J^*,\pi_0)^{-1}\cdot G_n\phi_i(\theta_0,\pi_0),
	\end{align*} where the basis vector $\|S(u)\|_2^2=O(1)$ by definition, hereby enabling bounded variance for all $J_n$. We can apply the Lindeberg-Feller similar as above to get
	\begin{align*}
		\Omega_{J,\tau}^{-1/2}\sqrt{n\kappa_{J_n}}(\hat\beta_n(\tau)-\beta_J^*(\tau))\xrightarrow{d}\mathcal{N}(0,\mathbf{I}), 
	\end{align*} where $\Omega_{J,\tau}:=S'(\tau)\Omega_{J,b}S(u)$. Since $\|\beta_J^*-\beta_0\|=O(1/J_n^{r+1})=o(1/\sqrt{n\kappa_{J_n}})$, the bias term $\hat\beta_n-\beta_0$ is dominated by $\hat\beta_n-\beta_J^*$, thus we have $\Omega_{J,\tau}^{-1/2}\sqrt{n\kappa_{J_n}}(\hat\beta_n(\tau)-\beta_0(\tau))\xrightarrow{d}\mathcal{N}(0,\mathbf{I})$.
\end{proof}

\subsection{Appendix D: Proofs of Results in Section \ref{subsec:asymp_np}}

\begin{proof}[\textbf{Proof of Theorem~\ref{thm5}}]
\label{pf_thm5} 
The proof of the asymptotic normality of $\hat\theta_n^{np}$ closely parallels that of $\hat\theta_n^{p}$. We therefore omit repetitive arguments and focus on the additional complications introduced by the nonparametric first stage. Denote $\hat\theta_n^{np}=\tilde\theta_n=(\tilde\beta_n(\cdot),\tilde\sigma,\tilde\gamma)$. Define the score function as $\psi_i(\theta,F)\equiv \psi(W_i;\theta,F):=\partial_\theta\log f(Y_i|X_i,F(X_i,Z_i);\theta)$. WLOG, let $d_x=1$. Since $\E_n[\psi_i(\tilde\theta_n,\hat{F})]=0$, by Taylor expansion around scalar $V_i=F_0(X_i,Z_i)$,
\begin{align*}
	0&=\underbrace{\frac{1}{\sqrt{n}}\cdot\sqrt{n}(\E_n-\E)\psi_i(\tilde\theta_n,F_0)}_{\mathbb{A}}
		+\underbrace{\E\psi(\tilde\theta_n,F_0)-\E\psi(\theta_J^*,F_0)+\E\psi(\theta_J^*,F_0)-\E\psi(\theta_0,F_0)}_{\mathbb{B}}\\
		&+\underbrace{\E_n[\partial_v\psi_i(\tilde\theta_n,F_0)(\hat{V}_i-V_i)]+O_p\left(\|\hat{V}_i-V_i\|^2\right)}_{\mathbb{C}},
\end{align*} where the expansion in $\mathbb{B}$ uses $\E\psi_i(\theta_0,F_0)=0$. Similar as in Proof~\ref{pf_thm3}, under the rate restrictions $J_n^{\frac{\lambda}{\lambda+1}}{\alpha_n}^{\frac{1}{2(\lambda+1)}}/\kappa_{J_n}\to0$ and $J_n^{r+1}/\sqrt{n}\to\infty$,
\begin{align*}
		\mathbb{A}&=\frac{1}{\sqrt{n}} G_n \bigl[ \psi_i(\hat{\theta}_n, F_0) \bigr] = \frac{1}{\sqrt{n}} (1 + o_p(1)) G_n \bigl[ \psi_i(\theta_0, F_0) \bigr]\\
		\mathbb{B}&=\mathcal{H}(\theta_J^*, F_0) (\tilde{b}_n - b_J^*, \tilde{\sigma} - \sigma_0, \tilde{\gamma} - \gamma_0)'+o_p(1),
\end{align*} with the sieve approximation error undersmoothed.

Consider the last term $\mathbb{C}$, by the Donsker of $\mathcal{G}$ as in Proof~\ref{pf_thm3},
\begin{align*}
	\mathbb{C}&=\E_n[\partial_v\psi_i(\tilde\theta_n,F_0)(\hat{V}_i-V_i)]+O_p\left(\alpha_n^2\right)\\
	&=\E_n[\partial_v\psi_i(\theta_0,F_0)(\hat{V}_i-V_i)]+\E_n[(\partial_v\psi_i(\tilde\theta_n,F_0)-\partial_v\psi_i(\theta_0,F_0))(\hat{V}_i-V_i)]+O_p\left(\alpha_n^2\right)\\
	&=\E_n[\partial_v\psi_i(\theta_0,F_0)(\hat{V}_i-V_i)]+O_p(\alpha_n(\sqrt{J_n/n}+\alpha_n)).
\end{align*} Decompose $\hat{V}_i-V_i$ into two components:
\begin{align*}
	\hat{V}_i-V_i&=\frac1n\sum_j p_i'\hat{Q}^{-}p_j\underbrace{\big( \1\{X_{1j}\le X_{1i}\}-F(X_{1i}|Z_j)\big)}_{:=\omega_{ij}}\\
	&+ \underbrace{\frac1n\sum_j p_i'\hat{Q}^{-}p_j \big(F(X_{1i}|Z_j)-p_j'a^K(X_{1i})\big)+(p_i'a^K(X_{1i})-F(X_{1i}|Z_i))}_{:= b(V_i)},
\end{align*} where $a^K(X_{1i})$ is the projection coefficient of $F(X_{1i}|Z_i)$ on $P_K(Z_i)$, and $b(V_i)=O_p(K_n^{1/2-d_1/d_z})$ by \citet{imbens2009identification} Lemma 11. Define the asymmetric kernel as $h(W_i,W_j):=\partial_v\psi_i(\theta_0,F_0)p_i'{Q}^{-}p_j\omega_{ij}$, by Hoeffding's decomposition,
\begin{align*}
	\mathbb{C}&=\frac{1}{n^2}\sum_{i}\sum_j h(W_i,W_j)+\frac{1}{n}\sum_i\partial_v\psi_i(\theta_0,F_0)b(V_i)+O_p(\alpha_n(\sqrt{J_n/n}+\alpha_n)),\\
	&=\frac{1}{n^2}\sum_{i,j}\Big\{\E[h(W_i,W_j)|W_j]+\tilde{h}(W_i,W_j)\Big\}+\E_n \partial_v\psi_i(\theta_0,F_0)b(V_i)+O_p(\alpha_n(\sqrt{J_n/n}+\alpha_n)),
\end{align*} where $\tilde{h}(W_i,W_j)$ represents the error after double projection, and the equality is from
\begin{align*}
	\E[h(W_i,W_j)]&=\E[\partial_v\psi_i(\theta_0,F_0)p_i'{Q}^{-}p_j\E[\omega_{ij}|W_i,Z_j]]=0\\
	\E[h(W_i,W_j)|W_i]&=\partial_v\psi_i(\theta_0,F_0)p_i'{Q}^{-}\E[p_j\E[\omega_{ij}|W_i,Z_j]|W_i]=0.
\end{align*} Let $\phi_{K_n}(\theta_0,W_j)=\E[h(W_i,W_j)|W_j]=\E[\partial_v\psi_i(\theta_0,F_0)p_i'{Q}^{-}p_j\omega_{ij}|W_j]$. Given Assumptions~\ref{parpro} and \ref{firstasym_np}, $\E\|\tilde{h}_{i,j}\|^2\le \E\|h_{i,j}\|^2\le C\cdot \E[(P_i'Q^{-}P_j)^2]=C\cdot \zeta_0(K_n)^2 K_n$ for some constant $C$, thus $\sqrt{n}\cdot \frac{1}{n^2}\sum_{i,j}\tilde{h}_{i,j}=O_p(\sqrt{\frac{\zeta_0(K_n)^2K_n}{n}})=o_p(1)$. Consequently,
\begin{align*}
	\sqrt{n}\mathbb{C}&=\sqrt{n}\Big(\frac{1}{n^2}\sum_{i,j}\phi_{K_n}(\theta_0,W_j)\Big)+o_p(1)+\sqrt{n}(\E_n-\E)\partial_v\psi_i(\theta_0,F_0)\cdot b(V_i)\\
	&+\sqrt{n}\E[\partial_v\psi_i(\theta_0,F_0)\cdot b(V_i)]+O_p(\alpha_n (\sqrt{J_n}+\alpha_n\sqrt{n}))\\
	&=G_n\phi_{K_n}(\theta_0,W_j)+\sqrt{n}\E[\partial_v\psi_i(\theta_0,F_0)b(V_i)]+o_p(1),
\end{align*} where first term on the right hand side is asymptotically Gaussian by Lindeberg-Feller CLT, while the second term $G_n\partial_v\psi_i b(V_i)$ is degenerate given that $Var(\partial_v\psi b)=O(\|b\|^2)=o(1)$, and the last term is dominated by the penultimate one. Define $\mu_J:=\mathcal{H}(\theta_J^*, F_0)^{-1}\E[\partial_v\psi_i(\theta_0,F_0) b(V_i)]$, with $\mu_{J,\cdot}$ be the corresponding subvectors. Combining $\mathbb{A}$, $\mathbb{B}$ and $\mathbb{C}$, we arrive at
\begin{align*}
	-\sqrt{n}\mathcal{H}(\theta_J^*, F_0) (\tilde{b}_n - b_J^*, \tilde{\sigma} - \sigma_0, \tilde{\gamma} - \gamma_0)'-\sqrt{n}\mathcal{H}(\theta_J^*, F_0)\mu_J = G_n \bigl[ \psi_i(\theta_0, F_0)+\phi_{K_n}(\theta_0,W_j) \bigr]
\\
	-\sqrt{n\kappa_{J_n}}\begin{pmatrix}
		\tilde{b}_n - b_J^*-\mu_{J,\beta}\\
		\tilde{\sigma} - \sigma_0-\mu_{J,\sigma}\\
		\tilde{\gamma} - \gamma_0-\mu_{J,\gamma}
	\end{pmatrix}
=\sqrt{\kappa_{J_n}}\mathcal{H}(\theta_J^*,F_0)^{-1}\Big[G_n \bigl[ \psi_i(\theta_0, F_0)+\phi_{K_n}(\theta_0,W_j) \bigr] \Big]
\end{align*} Since $\|\mu_J\|=O_p(K_n^{1/2-d_1/d_z}/\kappa_{J_n})$, if $\sqrt{\frac{n}{\kappa_{J_n}}}K_n^{\frac12-\frac{d_1}{d_z}}\to0$, then the bias term is asymptotically negligible. Variance of the right hand side is then
\begin{align*}
	\Omega_J:= \kappa_{J_n}\mathcal{H}(\theta_J^*, F_0)^{-1}\Big[\mathcal{I}(\theta_0,F_0)+\E[\phi_K\phi_K']+2\E[\psi(\theta_0,F_0)\phi_K'(\theta_0,W_j)]] \Big]\mathcal{H}(\theta_J^*, F_0)^{-1}.
\end{align*} Given the assumptions of first-stage series basis, the conclusion follows using arguments similar to Theorem~\ref{thm5} proof. One can analogously show that $\Omega_J$ has bounded eigenvalues, and pointwise asymptotic normality holds by the Lindeberg-Feller CLT.
\end{proof}

\subsection{Appendix E: Validity of Bootstrap}
\label{subsec:bootstrap}

We present the bootstrap validity result for the parametric first stage. Specifically, the bootstrap estimators $(\hat\beta_n^b,\hat\sigma^b,\hat\gamma^b)$, centered around $(\hat\beta_n,\hat\sigma,\hat\gamma)$, can be used to approximate the distribution of $(\hat\beta_n(\tau)-\beta_0(\tau),\hat\sigma-\sigma_0,\hat\gamma-\gamma_0)$. The validity of the bootstrap under an undersmoothed series first stage follows similarly given Theorem~\ref{thm5}. 

\begin{lem}\label{lem3}
	\textbf{Bootstrap Validity}\\
	Suppose the conditions of Theorem~\ref{thm3} hold, such that $J_n^{\frac{\lambda}{\lambda+1}}n^{\frac{-1}{4(\lambda+1)}}/\kappa_{J_n}\to0$, $J_n^{r+1}\kappa_{J_n}\to\infty$, and $J_n^{r+1}/\sqrt{n}\to\infty$, then for any fixed $\tau$,
	\begin{align*}
		\sqrt{n\kappa_{J_n}}\Omega^{-1/2}_{J,\tau}(\hat\beta_n^b(\tau)-\hat\beta_n(\tau))\xrightarrow{d}\mathcal{N}(0,\mathbf{I}_{d_x})\\
		\sqrt{n \kappa_{J_n}} \, \Omega_{J,\sigma}^{-1/2} (\hat{\sigma}^b-\hat{\sigma}) \xrightarrow{d} \mathcal{N}(0, \mathbf{I}_{d_\sigma}), \\
\sqrt{n \kappa_{J_n}} \, \Omega_{J,\gamma}^{-1/2} (\hat{\gamma}^b - \hat{\gamma}) \xrightarrow{d} \mathcal{N}(0, \mathbf{I}_{d_\gamma}),
	\end{align*} where $\Omega_{J,\cdot}$ are the same as in Theorem~\ref{thm3}.
\end{lem}

\begin{proof}[\textbf{Proof of Lemma~\ref{lem3}}]
\label{pf_lem3}
	Let $\{\xi_i\}_{i=1}^n$ be i.i.d. positive random weights satisfying $\E(\xi_i)=1$, $0\le Var(\xi_i)=\sigma_\xi^2<\infty$, and independent of the data. For example, one may take $(\xi_1,\dots,\xi_n)\sim \text{Multinomial}(n,\frac1n,\dots,\frac1n)$. Define the weighted bootstrap empirical measure as $\E_n^b\psi_i=\frac1n\sum_{i=1}^n \xi_i\psi_i$. Standard bootstrap results imply that the OLS estimator satisfies $(\hat\pi^b-\hat\pi)=O_p(n^{-1/2})$. Using arguments analogous to those in Theorem~\ref{thm2}, we obtain the consistency of $\hat\theta^b$. In particular,
	\begin{align*}
		\sup_{\theta\in\Theta_n}|\E_n[(\xi_i-1)l(W_i;\theta,\hat\pi^b)]-\E[(\xi_i-1)l(W_i;\theta,\hat\pi^b)]|=o_p(1).
	\end{align*} Since $\E[(\xi_i-1)l(W_i;\theta,\pi)]=0$, it follows that
	\begin{align*}
		\sup_{\theta\in\Theta_n}|\E_n^b l(W_i;\theta,\hat\pi^b)-\E_n l(W_i;\theta,\hat\pi^b)|=o_p(1),\\
		\sup_{\theta\in\Theta_n}|\E_n l(W_i;\theta,\hat\pi^b)-\E l(W_i;\theta,\pi_0)|=o_p(1),
	\end{align*} where the second equality uses the consistency of $\hat\pi^b$. The $L_2$ convergence rate of $(\hat\beta^b,\hat\sigma^b,\hat\gamma^b)$ likewise parallel those established in Lemma~\ref{lem1}. Given $P$-Donsker of $\mathcal{G}$,
	\begin{align*}
		\E_n^b\psi(\hat\theta^b_n,\hat\pi^b)-\E_n\psi_i(\hat\theta^b_n,\hat\pi^b)=\E_n^b\psi(\hat\theta_n,\hat\pi^b)-\E_n\psi_i(\hat\theta_n,\hat\pi^b)+o_p(n^{-1/2})\\
		\E_n\psi_i(\hat\theta_n,\hat\pi)-\E_n\psi_i(\hat\theta^b_n,\hat\pi^b)=\E_n^b\psi(\hat\theta_n,\hat\pi^b)-\E_n\psi_i(\hat\theta_n,\hat\pi^b)+o_p(n^{-1/2}),
	\end{align*} where the second equality is given by $\E_n^b\psi(\hat\theta^b_n,\hat\pi^b)=\E_n\psi_i(\hat\theta_n,\hat\pi)=0$. The RHS becomes 
	\begin{align*}
		\frac{1}{\sqrt{n}} \: (1+o_p(1))\: G_n\big[(\xi_i-1)\psi(\theta_0,\pi_0)\big]+o_p(n^{-1/2}),
	\end{align*} converging weakly to a Gaussian limit with covariance $\asymp \I(\theta_0,\pi_0)$. Meanwhile, the left-hand side admits the expansion:
	\begin{gather*}
		-\E_n \partial_\theta \psi_i(\hat\theta_n,\hat\pi)(\hat\theta^b_n-\hat\theta_n)- \E_n \partial_\pi \psi_i(\hat\theta_n,\hat\pi)(\hat\pi^b-\hat\pi)+o_p(n^{-1/2})\\
		-\E_n\partial_\theta \psi_i(\theta_J^*,\pi_0)(\hat\theta^b_n-\hat\theta_n)-\E_n \partial_\pi \psi_i(\hat\theta_n,\pi_0)(\hat\pi-\pi_0)+o_p(n^{-1/2})
	\end{gather*} Combining RHS and LHS, along with consistency of $\hat\theta_n$,
	\begin{align*}
		-\sqrt{n}\H(\theta_J^*,\pi_0)(\hat\theta^b_n-\hat\theta_n)= G_n\big[\psi(\theta_0,\pi_0)\big] +G_\pi\cdot \E[Z_iZ_i']^{-1}\cdot \frac{1
	}{\sqrt{n}}\sum_{i=1}^n Z_i\eta_i.
	\end{align*} Hence, the bootstrap estimator satisfies the same asymptotic expansion as in Theorem~\ref{thm3}, which establishes the desired result.
\end{proof}

\subsection{Appendix F: Multiple Endogenous Regressors}
\label{subsec:multi}

Suppose $X=(W', Z_1')'$, where $W$ is a vector of continuous endogenous regressors with $\dim(W)=K$, $1\le K\le d_x$. Let $Z=(Z_1',Z_2')'$, with exogenous variables $Z_1$ and excluded instruments $Z_2$. Assume $W=(X_1,\dots,X_K)'$ is generated by a multivariate triangular system:
\begin{align*}
	X_1&=h_1(Z,\upsilon_1)\\
	X_2&=h_2(Z,X_1,\upsilon_2)\\
	& \dots\\
	X_K&=h_K(Z,X_1,\dots,X_{K-1},\upsilon_K),
\end{align*} where each $h_k(\cdot,\upsilon_k)$ is strictly increasing in $\upsilon_k$. Let $\Upsilon=(\upsilon_1,\dots,\upsilon_K)'$, and assume $Z\ind (U,\Upsilon)$ and $\upsilon_i\ind \upsilon_j$ for all $i\neq j$. Namely, the instrument $Z$ is independent of the latent rank $U$ and all reduced-form disturbances, while each endogenous regressor affects $U$ through a distinct channel. By the Rosenblatt transform and the identification method in \citep{imbens2009identification}, we obtain a vector of control variables $\mathbf{V}=(V_1,\dots,V_K)$:\begin{align*}
	V_1&= F_{X_1|Z}(X_1|Z)=F_{\upsilon_1}(\upsilon_1)\\
	V_2&= F_{X_2|X_1,Z}(X_2|X_1,Z)=F_{\upsilon_2}(\upsilon_2)\\
	&\dots\\
	V_K &= F_{X_K|X_1,\dots,X_{K-1}}(X_k|X_1,\dots,X_{K-1},Z)=F_{\upsilon_k}(\upsilon_k),
\end{align*} each having a uniform marginal distribution. One can show, following Theorem 1 of \citep{imbens2009identification}, that $X\ind U\mid \mathbf{V}$. Although this framework does not require a rank condition, the common support assumption still necessitates that the excluded instruments provide sufficient variation.

\end{document}